\documentclass[12pt]{article}
\usepackage{epsfig}
\usepackage{graphicx,epstopdf,epsfig}
\usepackage{caption}
\usepackage{multirow}
\usepackage{booktabs}
\usepackage{tikz}
\usepackage{subcaption}
\usepackage{amsmath,amsthm,verbatim,amssymb,mathrsfs}
\usepackage{verbatim,color,amsmath}
\usepackage{algorithm, algorithmic}
\usepackage{hyperref}
\hypersetup{hypertex=true,
            colorlinks=true,
            linkcolor=blue,
            anchorcolor=blue,
            citecolor=blue}
\usepackage{enumerate}
\usepackage[authoryear, round]{natbib}
\usepackage{setspace}

\date{}
\usepackage{geometry}
\numberwithin{equation}{section}
\newtheorem{theorem}{Theorem}[section]
\newtheorem{lemma}{Lemma}[section]
\newtheorem{lemma*}{Lemma}
\newtheorem{remark}{Remark}[section]
\newtheorem{remark*}{Remark}

\newtheorem{proposition}{Proposition}[section]

\newcommand{\bZ}{{\bf Z}}
\newcommand{\bA}{{\bf A}}
\newcommand{\bB}{{\bf B}}
\newcommand{\bI}{\mbox{\bf I}}

\newcommand{\balpha}{\boldsymbol{\alpha}}
\def\b0{\mathbf{0}}
\def\bI{\mathbf{I}}
\def\Ebb{\mathbb{E}}
\def \mx {\mathbf{x}}
\def \my {\mathbf{y}}
\def \mv {\mathbf{v}}
\def \mz {\mathbf{z}}

\begin{document}
\title{\large \textbf{Schr{\"o}dinger-F{\"o}llmer Sampler:
 Sampling without Ergodicity}}

\author{
{\small Jian Huang}
\thanks{Department of Statistics and Actuarial Science, University of Iowa, Iowa City,  IA, USA.
Email: jian-huang@uiowa.edu}
\and
{\small Yuling Jiao}
\thanks{
School of Mathematics and Statistics, Wuhan University, Wuhan, China.
Email: yulingjiaomath@whu.edu.cn}
\and
{\small Lican Kang}
\thanks{
School of Mathematics and Statistics, Wuhan University, Wuhan, China.
Email: kanglican@whu.edu.cn}
\and
{\small Xu Liao}
\thanks{Center for Quantitative Medicine Duke-NUS Medical School, Singapore.
Email: liaoxu@u.duke.nus.edu}
\and
{\small Jin Liu}
\thanks{Center for Quantitative Medicine Duke-NUS Medical School, Singapore.
Email: jin.liu@duke-nus.edu.sg}
\and
{\small Yanyan Liu}
\thanks{School of Mathematics and Statistics, Wuhan University, Wuhan, China.
Email: liuyy@whu.edu.cn}
}
\maketitle
\begin{abstract}
Sampling from probability distributions is an important problem in statistics and machine learning, specially in Bayesian inference when integration with respect to  posterior distribution is intractable and sampling from the posterior is the only viable option for inference.  In this paper, we propose Schr\"{o}dinger-F\"{o}llmer sampler (SFS),  a novel approach to sampling from possibly unnormalized distributions. The proposed SFS is based on the  Schr\"{o}dinger-F\"{o}llmer diffusion process on the unit interval with a time-dependent drift term, which transports the degenerate distribution at time zero to the target distribution at time one. Compared with the existing Markov chain Monte Carlo samplers that require ergodicity, SFS does not need to have the property of ergodicity.  Computationally, SFS can be easily implemented using the Euler-Maruyama discretization. In theoretical analysis, we establish non-asymptotic error bounds for the sampling distribution of SFS in the Wasserstein distance under reasonable conditions. We conduct numerical experiments to evaluate the performance of SFS and demonstrate that it is able to generate samples with better quality than several existing methods.

\vspace{0.5cm} \noindent{\bf KEY WORDS}: Euler-Maruyama discretization, Non-asymptotic error bound,  Schr\"{o}dinger bridge,  Unnormalized distribution,
Wasserstein distance.
\end{abstract}

\section{Introduction}\label{introduction}
Sampling from a probability distribution
is a fundamental problem in statistics and machine learning. For example,  the ability to efficiently sample from an unnormalized posterior distribution is crucial to the success of Bayesian inference. Many sampling approaches have been developed in the literature. In particular, there is a large body of work on the Markov Chain Monte Carlo (MCMC)  methods, including the celebrated Metropolis-Hastings (MH) algorithm \citep{metropolis1953equation,hastings1970monte,tierney1994markov}, the Gibbs sampler \citep{geman1984, gelfand1990},
the Langevin algorithm \citep{roberts1996exponential,dalalyan2017theoretical,durmus2017nonasymptotic},
the bouncy particle sampler \citep{peters2012rejection,bouchard2018bouncy}, and the zig-zag sampler \citep{bierkens2019zig}, among others, see \citep{martin2020computing,changye2020markov,dunson2020hastings,brooks2011handbook}  and the references therein.

Among these methods,  the Langevin sampler based on the Euler-Maruyama discretization of Langevin diffusion has received much attention recently. The Langevin diffusion reads
\begin{align}\label{lange}
dL_t=-\nabla V(L_t) dt+\sqrt{2}\, dB_t,
\end{align}
where $-\nabla V(\cdot) $ is the drift term and $\{B_t\}_{t\geq0}$ is a standard $p$-dimensional Brownian motion process. Under suitable conditions, the Langevin diffusion process $\{L_t\}_{t\geq0}$ in \eqref{lange} admits an invariant distribution $\mu(x)=\exp (-V(x))/C, x \in \mathbb{R}^p$  with normalizing constant $C>0$   \citep{bakry2008rate,cattiaux2009trends}.
Nice convergence properties of the Langevin sampler under the strongly convex potential assumption have been established by several authors
\citep{durmus2016high-dimensional,
durmus2016sampling,dalalyan2017further,
dalalyan2017theoretical,cheng2018convergence,dalalyan2019user-friendly}.
Furthermore, the strongly convex potential assumption can be replaced by
different conditions to guarantee the log-Sobolev inequality for the target distribution,  including the dissipativity condition for the drift term \citep{raginsky2017non,mou2019improved,zhang2019nonasymptotic} and  the local convexity condition for the potential function outside a ball \citep{durmus2017nonasymptotic,cheng2018sharp,ma2019sampling,bou2020coupling}.

Although tremendous progress has been  made in the past decades,
it remains a challenging task to sample from distributions with multiple modes or distributions in high dimensions \citep{dunson2020hastings}.
Even for the one-dimensional Gaussian mixture model $0.5 N(-1,\sigma^2)+0.5N(1,\sigma^2)$,
the optimally tuned Hamiltonian Monte Carlo and the random walk Metropolis algorithms have the mixing time proportional to $\exp({1}/{(2\sigma^2)})$ \citep{mangoubi2018does,dunson2020hastings},
which will blow up exponentially as $\sigma$ decreases to $0$.
The constant in the log-Sobolev inequality may depend on the dimensionality exponentially \citep{menz2014poincare,wang2009log,hale2010asymptotic,raginsky2017non},
indicating that the efficiency of Langevin sampler may suffer from the curse of dimensionality when the ambient dimension $p$ is high.

In this paper, we propose Schr\"{o}dinger-F\"{o}llmer sampler (SFS),  a novel sampling approach without requiring the property of ergodicity.
SFS is based on the Schr\"{o}dinger-F\"{o}llmer diffusion
\begin{align}\label{sde}
\mathrm{d} X_{t}=b(X_{t}, t) \mathrm{d} t+\mathrm{d} B_{t}, \ X_0= 0, \  t \in[0,1],
\end{align}
where the function $b: \mathbb{R}^p \times [0,1]\to \mathbb{R}^p$ is a time-varying drift term determined by the target distribution. The specific form of $b$ is given by \eqref{driftbb} below.
According to \citet{leonard2014survey} and \citet{eldan2020}, the process  $\{X_t\}_{t\in [0,1]}$ (\ref{sde}) was first formulated by F\"{o}llmer \citep{follmer1985, follmer1986, follmer1988} when studying the Schr\"{o}dinger bridge problem \citep{schrodinger1932theorie}.

The law of  $\{X_t\}_{t\in [0,1]}$ in (\ref{sde}) minimizes the relative entropy with respect to the Wiener measure among all processes with laws interpolating $\delta_{0}$ (the degenerate distribution at $X_0=0$) and the target distribution $\mu$
\citep{jamison1975markov, dai1991stochastic, leonard2014survey}.
Therefore, the diffusion process (\ref{sde}) can be used to sample from the target distribution $\mu$ by transporting the initial degenerate distribution at $t=0$ to the target  $\mu$ at $t=1$.
To numerically implement this sampling approach, we use the Euler-Maruyama method to discretize the diffusion process \eqref{sde}. The resulting discretized version of \eqref{sde} is
\begin{align}\label{eld}
Y_{t_{k+1}}=Y_{t_k}+\delta_k  b\left(Y_{t_{k}}, t_{k}\right)+\sqrt{\delta_k }\epsilon_{k+1},\
k=0,1,\ldots, K-1, \ K \ge 2,
\end{align}
where {\color{black}$K$} is the number of the grid points on $[0,1]$ with $0=t_{0}<t_1<\ldots<t_K=1$, $\delta_k=t_{k+1}-t_k$ is the $k$-th step size,  and $\{\epsilon_{k}\}_{k=1}^{K}$ are independent and identically distributed random vectors from $N(0, \bI_{p})$.
Based on \eqref{eld}, we can start from $Y_{t_0}=0$ and iteratively update this initial value to obtain a realization of the random vector $Y_{t_K}$, which is approximately distributed as the target distribution $\mu$ under suitable conditions.
For convenience, we shall refer to the proposed sampling method as the Schr\"{o}dinger-F\"{o}llmer sampler (SFS).

An important difference between SFS and existing MCMC methods is that ergodicity is not required for SFS to generate valid samples. This is due to the basic property of the Schr\"{o}dinger-F\"{o}llmer diffusion \eqref{sde} on the unit time interval $[0, 1]$, which transports the initial distribution $\delta_0$ at $t=0$ to the exact target distribution
$\mu$ at $t=1$. The sampling error of SFS is entirely due to the Euler-Maruyama discretization and the approximation of the drift term in the Schr\"{o}dinger-F\"{o}llmer diffusion. These two type of approximation errors can be made arbitrarily small under suitable conditions.

Our main contributions are as follows.

\begin{enumerate}[(i)]
\item We propose a novel sampling method SFS without assuming ergodicity.
SFS is based on Euler-Maruyama discretization to
the Schr{\"o}dinger-F{\"o}llmer diffusion. The proposed SFS also works for unnormalized distributions.


\item
We establish non-asymptotic bounds for the difference between the law of the samples generated via SFS and the target distribution $\mu$ in terms of the Wasserstein distance
under appropriate conditions.
When the drift term $b$ can be calculated in a closed form, for example the target $\mu$ is a finite mixture of Gaussians, we show that $$W_2(\mbox{Law} (Y_{t_K}),\mu)\leq
\mathcal{O}(\sqrt{p/K})$$ under some smoothness conditions on the drift term,
see Theorem \ref{th1}.
In the case when the drift term needs to be calculated via Monte Carlo approximation, we prove that
$$W_2(\mbox{Law} (\widetilde{Y}_{t_K}),\mu)\leq
\mathcal{O}(\sqrt{p}(1/\sqrt{K}+1/\sqrt{m}))$$
under the assumption that the potential  $U(x,t)$ in \eqref{Udef} below
is strongly convex with respect to $x$, where $m$ is the number of Gaussian samples used in the Monte Carlo approximation of $b$, see Theorem \ref{th3}.

\item We conduct  numerical experiments to evaluate the effectiveness of the proposed SFS and demonstrate that SFS performs better than several  existing MCMC methods for Gaussian mixture models and Bayesian logistic regression.
\end{enumerate}

The rest of the paper is organized as follows. In Section \ref{method} we recall  some background and present the proposed SFS in details.
In Section \ref{compare} we compare SFS with Langevin sampler when the target distribution is a standard normal distribution. In Section \ref{Theorey} we establish the non-asymptotic bounds on the Wasserstein distance between the distribution of the samples generated via SFS and the target one. In Section \ref{simulation} we conduct simulation studies to assess the performance of SFS. Concluding remarks are given in Section \ref{conclusion}. Proofs for all the theorems are provided in Appendix \ref{append}.

We end this section by introducing some notation used throughout the paper.
Denote $\mathcal{B}(\mathbb{R}^{p})$ as the Borel set of $\mathbb{R}^{p}$,
and let $\mathcal{P}\left(\mathbb{R}^{p}\right)$
be the collection of probability measures on $(\mathbb{R}^{p}, \mathcal{B}(\mathbb{R}^{p}))$. Denote the gradient of a smooth
function $\varphi(x), x \in \mathbb{R}^{p}$ by $\nabla \varphi(x)$.
Similarly, denote the  partial derivative with respect to $x$ of $\phi(x,t), (x, t) \in \mathbb{R}^p\times[0,1]$ by $\nabla_x \phi(x,t)$.
For symmetric matrices  $\bA,\bB\in\mathbb{R}^{p\times p}$,
$\bA>\bB$ means that $\bA-\bB$ is a positive definite matrix.
Let $\|\beta\|_{\ell}=(\sum_{i=1}^{p}|\beta_{i}|^{\ell})^{\frac{1}{\ell}}$ be the  $\ell$-norm of the vector $\beta=(\beta_1,\ldots,\beta_p)^{\top}\in \mathbb{R}^{p}$. Denote $\|X\|_{L_2}=(\Ebb\|X\|_2^2)^{1/2}$
as the $L_2$-norm of a random vector $X$.
Denote $\Ebb_{X}$ as the expectation over the random vector $X$.
\section{Schr{\"o}dinger-F{\"o}llmer sampler}\label{method}
In this section we first provide some background on the Schr{\"o}dinger-F{\"o}llmer diffusion. We then present the proposed SFS method based on
the Euler-Maruyama discretization of the Schr{\"o}dinger-F{\"o}llmer diffusion.

\subsection{Background on Schr{\"o}dinger-F{\"o}llmer diffusion}
We first recall the Schr{\"o}dinger bridge problem, then introduce the Schr{\"o}dinger-F{\"o}llmer diffusion, and at last calculate the closed form expression of the drift term in the scenario that the target distribution is  
a Gaussian mixture.

\subsubsection{Schr{\"o}dinger bridge problem}
Let $\Omega = C([0,1],\mathbb{R}^p)$ be the space of $\mathbb{R}^p$-valued continuous functions on the time interval $[0, 1]$. Denote $Z = (Z_t)_{t\in [0,1]}$ as the canonical process on $\Omega$, where $Z_t(\omega) = \omega_t$, $\omega = (\omega_s)_{s\in [0,1]}\in \Omega$.
The canonical $\sigma$-field on $\Omega$ is then generated as $\mathscr{F}  = \sigma(Z_t,t\in[0,1]) = \left\{\{\omega:(Z_t(\omega))_{t\in [0,1]}\in H\}:H\in\mathcal{B}(\mathbb{R}^p)\right\}$. Denote $\mathcal{P}(\Omega)$ as the space of probability measures on the path space $\Omega$, and $\mathbf{W}_{\mx}\in\mathcal{P}(\Omega)$ as the Wiener measure
whose initial marginal is $\delta_{\mx}$. The law of the reversible Brownian motion, is then defined as $\mathbf{P}= \int \mathbf{W}_{\mx}\mathrm{d}\mx$, which is an unbounded measure on $\Omega$. One can observe that, $\mathbf{P}$ has a marginal coinciding with the Lebesgue measure $\mathscr{L}$ at each $t$.
\citet{schrodinger1932theorie} studied the problem of finding the most likely random evolution between two probability distributions $\widetilde{\nu}, \widetilde{\mu} \in \mathcal{P}(\mathbb{R}^p)$.
This problem is referred to as the Schr\"{o}dinger bridge problem (SBP).
SBP can be further formulated as seeking a probability law on the path space that interpolates between $\widetilde{\nu}$ and $\widetilde{\mu}$, such that the probability law is close to the prior law of the Brownian diffusion with respect to the relative entropy \citep{jamison1975markov,leonard2014survey}, i.e.,  finding a path measure $\mathbf{Q}^* \in \mathcal{P}(\Omega)$ with marginal $\mathbf{Q}^*_{t} = (Z_t)_{\#}\mathbf{Q}^*=\mathbf{Q}^*\circ Z_t^{-1}, t\in [0,1]$ such that
$$\mathbf{Q}^* \in \arg \min \mathbb{D}_{\mathrm{KL}}(\mathbf{Q}||\mathbf{P}),$$ and  $$\mathbf{Q}_0 = \widetilde{\nu}, \mathbf{Q}_1 = \widetilde{\mu},$$ where
the relative entropy $\mathbb{D}_{\mathrm{KL}}(\mathbf{Q}||\mathbf{P}) = \int \log(\frac{d \mathbf{Q}}{d \mathbf{P}}) d \mathbf{Q} $ if $\mathbf{Q}\ll \mathbf{P}$ (i.e. $\mathbf{Q}$ is absolutely continuous w.r.t. $\mathbf{P}$), and $\mathbb{D}_{\mathrm{KL}}(\mathbf{Q}||\mathbf{P}) = \infty$ otherwise.
The following theorem characterizes the solution of SBP.

\begin{theorem}\label{th01}\citep{leonard2014survey}
If $\widetilde{\nu}, \widetilde{\mu} \ll \mathscr{L}$,  then SBP admits a unique solution $\mathbf{Q}^* = f^*(Z_0)g^*(Z_1)\mathbf{P}$, where
$f^*$ and $g^*$ are $\mathscr{L}$-measurable nonnegative  functions satisfying the  Schr\"{o}dinger system
$$\left\{\begin{array}{l}
f^*(\mx) \mathbb{E}_{\mathbf{P}}\left[g^*\left(Z_{1}\right) \mid Z_{0}=\mx\right]= \frac{d \widetilde{\nu}}{d\mathscr{L}}(\mx), \quad \mathscr{L}-a . e . \\
g^*(\my)  \mathbb{E}_{\mathbf{P}}\left[f^{*}\left(Z_{0}\right) \mid Z_{1}=\my\right]=\frac{d \widetilde{\mu}}{d\mathscr{L}}(\my),  \quad \mathscr{L}-a . e .
\end{array}\right.$$
Furthermore, the pair $(\mathbf{Q}^*_{t},\mv^*_{t})$ with $$\mv^*_{t}(\mx) = \nabla_{\mx} \log \mathbb{E}_{\mathbf{P}}\left[g^{*}\left(Z_{1}\right) \mid Z_{t}= \mx\right]$$ solves the minimum action problem
$$\min_{\mu_t,\mv_t} \int_{0}^{1} \mathbb{E}_{\mz\sim \mu_t}[\|\mv_t(\mz)\|^2] \mathrm{d} t$$ s.t.
$$\left\{\begin{array}{l}
\partial_{t}\mu_t = -\nabla \cdot(\mu_t \mv_t) +  \frac{\Delta \mu_t}{2}, \quad \text { on }(0,1) \times \mathbb{R}^{p} \\
\mu_{0}=\widetilde{\nu}, \mu_{1}=\widetilde{\mu}.
\end{array}\right.
$$
\end{theorem}
Let  $K(s, \mx, t, \my) = [2\pi(t-s)]^{-p/2}\exp\left(-\frac{\|\mx - \my\|^2}{2(t-s)}\right)$ be the transition density of the Wiener process,
$\widetilde{q}(\mx)$ and $\widetilde{p}(\my)$ be
the density of $\widetilde{\nu}$ and $\widetilde{\mu}$, respectively.  Denote
$$f_{0}(\mx) = f^*(\mx), \ \ g_{1}(\my) = g^*(\my),$$
$${f_{1}}(\my) = \mathbb{E}_{\mathbf{P}}\left[f^{*}\left(Z_{0}\right) \mid Z_{1}=\my\right] = \int K(0, \mx, 1, \my)f_{0}(\mx) d \mx,$$
$${g_{0}}(\mx)= \mathbb{E}_{\mathbf{P}}\left[g^*\left(Z_{1}\right) \mid Z_{0}=\mx\right] = \int K(0, \mx, 1, \my)g_{1}(\my) d \my.$$
Then, the  Schr\"{o}dinger system  in Theorem \ref{th01} can also be characterized by
\begin{equation}\label{sbs}
\widetilde{q}(\mx) = f_0(\mx) {g_{0}}(\mx), \ \  \widetilde{p}(\my)=  {f_{1}}(\my)g_1(\my)
\end{equation}
with the following forward and backward time harmonic equations  \citep{chen2020stochastic},
 $$\left\{\begin{array}{l}
\partial_t f_t(\mx) = \frac{\Delta}{2} f_t(\mx),  \\
\partial_t g_t(\mx) = -\frac{\Delta}{2} g_t(\mx),
\end{array}\right. \quad \text { on }(0,1) \times \mathbb{R}^{p}.
$$

Let $q_t$ denote marginal density of $\mathbf{Q}_t^{*}$, i.e., $q_t(\mx) = \frac{d\mathbf{Q}_t^{*}}{d\mathscr{L}}(\mx)$,  then it can be represented by the product of  $g_t$  and $f_t$ \citep{chen2020stochastic}. Let  $\mathcal{V}$ consist of admissible Markov controls with finite energy. Then,  the vector field
\begin{equation}\label{drift}
\begin{aligned}
\mv^*_{t}=\nabla_{\mx}\log g_t(\mx)
= \nabla_{\mx}\log  \int K(t, \mx, 1, \my)g_1(\my) \mathrm{d} \my
\end{aligned}
\end{equation}
solves the following stochastic control problem.
\begin{theorem}\label{th02}\citep{dai1991stochastic}
$$\mathbf{v}^*_{t}(\mx)\in \arg\min_{\mathbf{v} \in \mathcal{V}}\mathbb{E}\left[\int_0^1\frac{1}{2}\|\mathbf{v}_t\|^2\mathrm{d}t\right]$$
s.t.
\begin{equation}\label{sdeb}
\left\{\begin{array}{l}
\mathrm{d}\mx_t = \mathbf{v}_t \mathrm{d}t+\mathrm{d} B_t, \\
\mx_0\sim \widetilde{q}(\mx),\quad \mx_1\sim \widetilde{p}(\mx).
\end{array}\right.
\end{equation}
\end{theorem}
According to Theorem \ref{th02}, the dynamics determined by the SDE in \eqref{sdeb} with a time-varying drift term $\mathbf{v}^*_{t}$ in \eqref{drift} will drive the particles sampled from the initial distribution $\widetilde{\nu}$ to evolve to the particles drawn from the target distribution $\widetilde{\mu}$ on the unit time interval. This nice property is what we need in designing samplers: we can sample from the underlying target distribution $\widetilde{\mu}$ via pushing forward a simple reference distribution $\widetilde{\nu}$.
In particular, if we take the initial distribution $\widetilde{\nu}$ to be $\delta_0$, the degenerate distribution at $0$, then the Schr\"{o}dinger-F\"{o}llmer diffusion process \eqref{sch-equation} defined below is a solution to \eqref{sdeb}, i.e., it
will transport $\delta_0$ to the target distribution.

\subsubsection{Schr\"{o}dinger-F\"{o}llmer diffusion process}
Let $\mu \in \mathcal{P}\left(\mathbb{R}^{p}\right)$ denote the target distribution of interest. 
Suppose $\mu$ is absolutely continuous with respect to the $p$-dimensional standard Gaussian distribution $N(0,\bI_p).$
Let $f$ denote the Radon-Nikodym derivative of $\mu$ with respect to $N(0, \bI_p)$, or the ratio of the density of $\mu$ over the density of $N(0, \bI_p),$ i.e.,
\begin{equation}\label{drb}
 f(x)=\frac{d\mu}{dN(0,\bI_p)}(x), \ x \in \mathbb{R}^p.
 \end{equation}
Let $Q_t$ be the heat semigroup defined by
\begin{equation}\label{qdr}
Q_{t} f(x)=\Ebb_{Z \sim N(0,\bI_p)}[f(x+\sqrt{t} Z)],~ t \in [0, 1].
\end{equation}
The Schr\"{o}dinger-F\"{o}llmer diffusion process $\{X_t\}_{t\in [0,1]}$
is defined as \citep{follmer1985, follmer1986, follmer1988}
\begin{align}\label{sch-equation}
\mathrm{d} X_{t}=-\nabla_x U\left(X_{t}, t\right) \mathrm{d} t+\mathrm{d} B_{t},  \ X_{0}=0,  \  t \in [0,1],
\end{align}
where $U$ is the potential given by
\begin{equation}
\label{Udef}
U(x, t)=-\log Q_{1-t} f(x).
\end{equation}
This process $\{X_t\}_{t\in [0,1]}$ defined by \eqref{sch-equation} is a solution to \eqref{sdeb} with
 $\tilde{\nu} = \delta_{0}$, $\tilde{\mu} = \mu$, and $\mathbf{v}_{t}(x) = -\nabla_x U(x,t)$ \citep{dai1991stochastic, lehec2013, eldan2020}.
For notational convenience, we denote the drift term of the SDE \eqref{sch-equation} by
\begin{equation}\label{driftbb}
b(x,t)\equiv -\nabla_x U(x,t)=\frac{\Ebb_{Z}[\nabla f(x+\sqrt{1-t}Z)]}{\Ebb_{Z}[f(x)+\sqrt{1-t} Z)]},  \  x \in \mathbb{R}^p, t \in [0, 1],
\end{equation}
where $Z \sim N(0, \bI_p)$.

To ensure that the SDE \eqref{sch-equation} admits a unique strong solution, we assume that the drift term $b$ satisfies a linear growth condition and a Lipschitz continuity condition \citep{revuz2013continuous,pavliotis2014stochastic},
\begin{align}\label{cond1}
\|b(x,t)\|_2^2\leq C_0(1+\|x\|_2^2) \tag{C1}, \ x \in \mathbb{R}^p, t \in [0,1]
\end{align}
and
\begin{align}\label{cond2}
\|b(x,t)-b(y,t)\|_2\leq C_1 \|x-y\|_2 \tag{C2}, \ x, y \in \mathbb{R}^p, t \in [0,1],
\end{align}
where $C_0$ and $C_1$ are finite positive constants.
\begin{proposition}\label{SBP}
Assume \eqref{cond1} and \eqref{cond2} hold, then the Schr{\"o}dinger-F\"{o}llmer
SDE  \eqref{sch-equation}
has a unique strong solution $\{X_t\}_{t\in[0,1]}$ with $X_0 \sim \delta_0$ and  $X_{1} \sim \mu$.
\end{proposition}
\begin{remark}\label{Re1}
\item[(i)]  Proposition  \ref{SBP} is a known property of the Schr{\"o}dinger-F\"{o}llmer process
\citep{dai1991stochastic,lehec2013,tzen2019theoretical,eldan2020}.
See also the review \citep{leonard2014survey} for additional discussions and historical account on the Schr\"{o}dinger problem.
\item[(ii)] The drift term $b(x,t)$ is scale-invariant with respect to $f$ in the sense that
$b(x, t)=\nabla \log Q_{1-t} Cf(x), \forall C>0$. Therefore, the Schr{\"o}dinger-F\"{o}llmer diffusion can be used for sampling from an unnormalized distribution $\mu$,  that is,  the normalizing constant of  $\mu$ does not need to be known.
\item[(iii)]
If  $f$ and $\nabla f$ are Lipschitz continuous, and $f$ has a lower bound strictly greater than 0, then we can easily deduce that assumptions \eqref{cond1} and \eqref{cond2} hold.
 When the target distribution $\mu$ is a Gaussian mixture distribution \eqref{gaumix} below,
$f$ and $\nabla f$ are Lipschitz continuous if $\bI_p>\Sigma_i, i=1,\ldots,\kappa$.
\end{remark}

Suppose that the target distribution has a density function with respect to the Lebesgue measure $\mathscr{L}$  on $(\mathbb{R}^p,\mathcal{B}(\mathbb{R}^p))$.
Let $\mu$ also denote the density function with a normalizing constant $C>0$.
Without loss of generality, we write
\begin{align}\label{tdensity}
\mu(x)=\frac{1}{C} \exp (-V(x)), \   x  \in \mathbb{R}^p,
\end{align}
where $V$ has a known form, but $C$ may be unknown.
The Radon-Nikodym derivative of $\mu$ with respect to $N(0, \bI_p)$ can be written as $f(x)= C^{-1}{(2\pi)^{p/2}} g(x)$,
where
\begin{align}\label{gf}
g(x) = \exp \left(-V(x)+\frac{1}{2}\|x\|^{2} \right), \ x \in \mathbb{R}^p.
\end{align}
If  the potential $V$ in \eqref{tdensity} takes the form
\begin{align}\label{vfun}
V(x)=a_1x^{\top}\bA x+a_2\eta^{\top}x+a_3,
\end{align}
where $a_1,a_2,a_3 \in \mathbb{R}$ are constants, $\bA \in \mathbb{R}^{p\times p}$ is a positive definite matrix  and $\eta \in \mathbb{R}^{p}$ is a vector, then the closed form expression of $b$ can be computed.
Several common distributions are special cases.
For example,  $\mu$ simplifies to a uniform distribution if $a_1=a_2=0$ and its support is a bounded set in $\mathbb{R}^p$; it is an exponential distribution if $a_1=0$;  and it is a normal distribution for $a_1>0$.
Also, when the target distribution $\mu$ is a finite mixture of distributions with potential function given by \eqref{vfun}, the drift terms can be calculated explicitly.
Using widely used Gaussian mixture models as an illustrative example, we derive the explicit form expression of the corresponding drift terms in the following subsection.
\subsubsection{Gaussian mixture distributions}
Assume that the target distribution $\mu$ is a Gaussian mixture,  i.e.,
\begin{align}\label{gaumix}
\mu=\sum_{i=1}^{\kappa}\theta_iN(\alpha_i,\Sigma_i), \
\sum_{i=1}^{\kappa} \theta_i=1 \text{ and } 0\leq\theta_i\leq1, i=1,\ldots, \kappa,
\end{align}
where $\kappa$ is the number of mixture components, $N(\alpha_i,\Sigma_i)$ is the $i$th Gaussian component with mean $\alpha_i\in \mathbb{R}^p$ and covariance matrix $\Sigma_i\in \mathbb{R}^{p\times p}$.
Obviously, the target distribution $\mu$ in \eqref{gaumix} is absolutely continuous with respect to the $p$-dimensional standard Gaussian distribution $N(0,\bI_p)$.  The density ratio is $f=\sum_{i=1}^{\kappa} \theta_i f_i$, where $f_i=\frac{d N(\alpha_i,\Sigma_i)}{d N(0,\bI_p)}$ is the density ratio of $N(\alpha_i, \Sigma_i)$ over $N(0, \bI_p)$.
The drift term of  the
Schr{\"o}dinger-F{\"o}llman SDE \eqref{sch-equation} is
\begin{align}\label{dr}
b(x,t)=\frac{\sum_{i=1}^{\kappa} \theta_i \Ebb_{Z}\nabla f_i(x+\sqrt{1-t}Z)}{\sum_{i=1}^{\kappa} \theta_i \Ebb_{Z} f_i(x+\sqrt{1-t}Z)}, \ Z \sim N(0, \bI_p).
\end{align}
To obtain the expression of the drift term $b(x,t)$ in \eqref{dr}, we only need to derive the expressions of  $\Ebb_{Z}\nabla f_i(x+\sqrt{1-t}Z)$ and $\Ebb_{Z} f_i(x+\sqrt{1-t}Z)$, $i=1\ldots \kappa$.
Some tedious calculation shows that
\begin{align}
&
\Ebb_{Z}\nabla f_i(x+\sqrt{1-t}Z) \nonumber\\
&=\frac{\Sigma_i^{-1}\alpha_i+(\bI_p-\Sigma_i^{-1})[t\bI_p+(1-t)\Sigma_i^{-1}]^{-1}
[(1-t)\Sigma_i^{-1}\alpha_i+x]}
{|t\Sigma_i+(1-t)\bI_p|^{1/2}}g_i(x,t), \label{exp1} \\
&\Ebb_{Z} f_i(x+\sqrt{1-t}Z)=\frac{g_i(x,t)}{|t\Sigma_i+(1-t)\bI_p|^{1/2}} ,
\label{exp2}
\end{align}
where
\begin{align*}
g_i(x,t)&=\exp\left(\frac{1}{2-2t}\|(t\bI_p+(1-t)\Sigma_i^{-1})^{-1/2}((1-t)\Sigma_i^{-1}\alpha_i+x)\|_2^2\right)\\
&~~~~\times \exp\left(-\frac{1}{2}\alpha_i^T\Sigma_i^{-1}\alpha_i-\frac{1}{2-2t}\|x\|_2^2
\right).
\end{align*}
Therefore,  we can obtain an analytical expression of $b(x,t)$  by plugging the expressions \eqref{exp1}
and \eqref{exp2} into \eqref{dr}. See Appendix \ref{append} for details.
\subsection{SFS based on Euler-Maruyama discretization} \label{discretization}
Proposition  \ref{SBP} shows that we can start from $X_0=0$ and update the values of
$\{X_t: 0 < t \le 1\}$ according to the Sch\"{o}linger-F\"{o}llmer SDE \eqref{sch-equation} in continuous time, then the value $X_1$
 has the desired distributional property, that is,
$X_1 \sim \mu$.  Hence, to implement this sampling procedure computationally, we just need to discretize the continuous process.  We use the Euler-Maruyama discretization for the SDE \eqref{sch-equation} with a fixed step size.
Let $$t_{k}=k\cdot s, ~k=0,1, \ldots, K, ~\mbox{with}~ s=1 / K,$$
and set $Y_{t_{0}}=0$.
Then the Euler-Maruyama discretization of \eqref{sch-equation} has the form
\begin{align}\label{emd}
Y_{t_{k+1}}=Y_{t_k}+s  b(Y_{t_{k}}, t_{k})+\sqrt{s}\,\epsilon_{k+1},
~
k=0,1,\ldots, K-1,
\end{align}
where $\{\epsilon_{k}\}_{k=1}^{K}$ are independent and identically distributed random vectors from $N(0,\bI_{p})$ and
\begin{align}\label{driftb}
b(Y_{t_{k}},t_{k})=\frac{\Ebb_{Z}[\nabla f(Y_{t_{k}}+\sqrt{1-t_{k}}Z)]}{\Ebb_{Z}[f(Y_{t_{k}}+\sqrt{1-t_{k}} Z)]}, \ Z \sim N(0, \bI_p).
\end{align}
The main computational task involved in updating the Euler-Maruyama discretization \eqref{emd} is to compute the drift term $b$ defined in \eqref{driftb}. The following points are worth noting.
\begin{enumerate}[(a)]
\item Recall the Radon-Nikodym derivative
 $f(x)=\frac{d \mu}{d N(0,\bI_p)}(x).$
 The normalizing constant of $\mu$ cancels out from the numerator and the denominator of the drift term $b$. Therefore, SFS can sample from unnormalized distributions.

 \item It is generally intractable to calculate the drift term $b$  analytically when the target distribution $\mu$ has a complex structure.  Also, it involves the derivative $\nabla f$, which can have a complicated form and be difficult to compute.
\end{enumerate}
By Stein's lemma
\citep{stein1972, stein1986,
landsman2008stein},
we have
\[
\Ebb_{Z}[\nabla f(Y_{t_{k}}+\sqrt{1-t_{k}}Z)]=\frac{1}{\sqrt{1-t_{k}}}
\Ebb_{Z}[Z f(Y_{t_{k}}+\sqrt{1-t_{k}}Z)].
\]
This identity enables us to avoid the calculation of $\nabla f.$
The drift term can be rewritten as
\begin{align*}
b(Y_{t_{k}},t_{k})=
\frac{\Ebb_{Z}[Z f(Y_{t_{k}}+\sqrt{1-t_{k}}Z)]}{\Ebb_{Z}[f(Y_{t_{k}}+\sqrt{1-t_{k}} Z)]\cdot\sqrt{1-t_k}}.
\end{align*}
Since $g$ in \eqref{gf} is proportional to $f$ up to a multiplicative constant independent of $x$, we can express $b$ in terms of $g$ as
\begin{equation}\label{driftd1}
b(Y_{t_{k}},t_{k})=
\frac{\Ebb_{Z}[\nabla g(Y_{t_{k}}+\sqrt{1-t_{k}}Z)]}{\Ebb_{Z}[g(Y_{t_{k}}+\sqrt{1-t_{k}} Z)]},
\end{equation}
or
\begin{equation}\label{driftd2}
b(Y_{t_{k}},t_{k})=
\frac{\Ebb_{Z}[Z g(Y_{t_{k}}+\sqrt{1-t_{k}}Z)]}{\Ebb_{Z}[g(Y_{t_{k}}+\sqrt{1-t_{k}} Z)]\cdot\sqrt{1-t_k}}.
\end{equation}
This expression no longer involves any unknown constants.
The pseudocode for implementing \eqref{emd} is presented in Algorithm \ref{alg:1}.
\begin{algorithm}[H]
\caption{SFS  $\mu = \exp (-V(x))/C$.}
\label{alg:1}
	\begin{algorithmic}[1]
\STATE Input:  $V(x)$, $K$.  Initialize $s=1/K$, $Y_{t_0}=0$.
\FOR{$k= 0,1,\ldots, K-1$ }
\STATE Sample $\epsilon_{k+1}\sim N(0,\bI_{p})$,
\STATE Compute the drift term $b(Y_{t_{k}}, t_{k})$ by \eqref{driftd1} or \eqref{driftd2},
\STATE Update
$Y_{t_{k+1}}=Y_{t_{k}}+sb(Y_{t_{k}}, t_{k})+\sqrt{s} \epsilon_{k+1}$.
\ENDFOR
\STATE Output:  $\{Y_{t_k}\}_{k=1}^{K}$.
\end{algorithmic}
\end{algorithm}

However, unlike the case of Gaussian mixture distributions \eqref{gaumix} discussed earlier, the drift term $b$ generally does not have a closed form expression.
Fortunately, it can be easily calculated approximately up to any desired precision via Monte Carlo. Let $Z_1, \ldots, Z_{m}$ be i.i.d. $N(0, \bI_p)$, where $m\ge 1$ is sufficiently large. Based on \eqref{driftd1} and \eqref{driftd2}, we can approximate $b$ by
\begin{align}\label{drifte1}
\tilde{b}_m(Y_{t_{k}},t_{k})=
\frac{\frac{1}{m}\sum_{j=1}^m[\nabla g(Y_{t_{k}}+\sqrt{1-t_{k}}Z_j)]}{\frac{1}{m}\sum_{j=1}^{m} [g(Y_{t_{k}}+\sqrt{1-t_{k}} Z_j)]}, \ k=0, \ldots, K-1,
\end{align}
or
\begin{align}\label{drifte2}
\tilde{b}_m(Y_{t_{k}},t_{k})=
\frac{\frac{1}{m}\sum_{j=1}^m[Z_j g(Y_{t_{k}}+\sqrt{1-t_{k}}Z_j)]}{\frac{1}{m}\sum_{j=1}^{m} [g(Y_{t_{k}}+\sqrt{1-t_{k}} Z_j)]\cdot\sqrt{1-t_k}}, \ k=0, \ldots, K-1.
\end{align}
Then, the Euler-Maruyama discretization \eqref{emd} becomes
$$
\widetilde{Y}_{t_{k+1}}=\widetilde{Y}_{t_{k}}+s\tilde{b}_{m}(\widetilde{Y}_{t_{k}}, t_{k})
+\sqrt{s} \, \epsilon_{k+1}, k=0,1, \ldots, K-1,
$$
where $\{\epsilon_{k}\}_{k=1}^{K}$ are i.i.d. $N\left(0,\bI_{p}\right)$.
We present the pseudocode of SFS in Algorithm \ref{alg:2} below.
\begin{algorithm}[H]
	\caption{SFS for $\mu = \exp (-V(x))/C$  with  Monte Carlo estimation of the drift term}
    \label{alg:2}
	\begin{algorithmic}[1]
\STATE Input: $V(x)$, $m$, $K$.  Initialize $s=1/K$, $\widetilde{Y}_{t_0}=0$.
\FOR{$k= 0,1,\ldots, K-1$ }
\STATE Sample $\epsilon_{k}\sim N(0,\bI_{p})$,
\STATE Compute $\tilde{b}_{m}$ according to \eqref{drifte1} or \eqref{drifte2},
\STATE Update
$\widetilde{Y}_{t_{k+1}}=\widetilde{Y}_{t_{k}}+s\tilde{b}_{m}\left(\widetilde{Y}_{t_{k}}, t_{k}\right)+\sqrt{s}\, \epsilon_{k+1}$.
\ENDFOR
\STATE Output:  $\{\widetilde{Y}_{t_k}\}_{k=1}^{K}$
\end{algorithmic}
\end{algorithm}
\section{Comparison with sampling via Langevin diffusion }\label{compare}
There are several important differences between the Schr{\"o}dinger-F{\"o}llman  diffusion and the Langevin diffusion. First, the SDE \eqref{sch-equation} is defined on the finite time interval $[0,1]$, while the Langevin diffusion in \eqref{lange} is defined on the infinite time interval $[0,\infty)$. Second,  the drift term $b$ is determined by the Radon-Nikydom derivative of the target distribution with respect to the Gaussian distribution in \eqref{sch-equation} and  is time-varying; in comparison, the drift term $-\nabla V(\cdot)$ in \eqref{lange} is the gradient of the log target density and independent of time. Third and most importantly, $X_1$ of the Schr{\"o}dinger-F{\"o}llman diffusion in \eqref{sch-equation} is exactly distributed as the target distribution $\mu$, but the law of $L_t$ only converges to the target distribution as $t$ goes to infinity.

To further illustrate the difference between the Schr{\"o}dinger-F{\"o}llman diffusion and the Langevin diffusion, we consider the canonical case when the target distribution is the standard Gaussian $N(0,\bI_p)$.
\subsection{Standard normal  target: continuous solution}
 In this case,  $b(x,t)\equiv 0$ in the Schr{\"o}dinger-F{\"o}llman diffusion \eqref{sch-equation}.
Hence,  $X_t=B_t$ for $t \in [0, 1]$.
Therefore, $X_1$ is exactly Gaussian, i.e., $X_1\sim N(0,\bI_p)$.
In this scenario, the Langevin diffusion \eqref{lange} becomes
\begin{align}\label{ou}
dL_t=-L_t dt+\sqrt{2}dB_t,~t\geq 0.
\end{align}
The SDE \eqref{ou} defines an Ornstein-Uhlenbeck process with the transition probability density
$$
p_t(y)=\frac{1}{\sqrt{2\pi(1-\exp(-2t))}}\exp\left(-\frac{(y-y_0\exp(-t))^2}{2(1-\exp(-2t))}\right),
$$
if $L_0=y_0$, see \cite{pavliotis2014stochastic} for details.
The law of $L_t$ is not exactly normal but will converge to the standard Gaussian
distribution as $t \to \infty$, irrespective of the initial position $y_0$.

Therefore, in the case of standard Gaussian, the distribution of $X_1$ in the  Schr{\"o}dinger-F{\"o}llman SDE  \eqref{sch-equation} is exactly the same as the target distribution. In comparison, the distribution of $L_t$ in the Langevin diffusion only converges to the target distribution as $t\rightarrow \infty$.

\subsection{Standard normal target: Euler-Maruyama discretization}
Next, we compare Euler-Maruyama discretizations of  Schr{\"o}dinger-F{\"o}llman  SDE \eqref{sch-equation} and the Langevin SDE \eqref{ou} when the target distribution $\mu=N(0,\bI_p)$.
In this case, the Euler-Maruyama discretization \eqref{emd} of  Schr{\"o}dinger-F{\"o}llman  SDE \eqref{sch-equation} yields
 \begin{align*}
Y_{t_{k+1}}=Y_{t_k}+\sqrt{s}\epsilon_{k+1},
~ k=0,1,\ldots, K-1.
\end{align*}
Since $Y_{t_0}=0$, then $Y_{t_K}=\sqrt{s}\sum_{k=0}^{K-1}\epsilon_{k+1}$ is exactly distributed as  $N(0,\bI_p)$. That is, SFS is an exact sampler in finite  steps for any fixed step size.

In comparison, the Euler-Maruyama iterative sequence of \eqref{ou} is
\begin{align}\label{oulange}
\widetilde{L}_{t_{k+1}}=(1-h_k)\widetilde{L}_{t_{k}}+\sqrt{2h_k}\epsilon_{k+1}, ~ k=0,1,\ldots,
\end{align}
where $h_k=t_{k+1}-t_{k}$ is the step size. If we set $h_k=h$ to be the fixed step size and $\widetilde{L}_{t_0}=0$,
then $\widetilde{L}_{t_{k}}=\sqrt{2h}\sum_{i=0}^{k-1}(1-h)^{i}\epsilon_{k-i}$ is distributed as $N\left(0,\frac{2(1-(1-h)^{2k})}{2-h}\bI_p\right)$.
So the law of $\widetilde{L}_{t_{k}}$ will converge to $N\left(0,\frac{2}{2-h}\bI_p\right)$ as $k \to \infty$ for any given $0<h<1$. But this limit distribution $N\left(0,\frac{2}{2-h}\bI_p\right)$ is still not equal to $N(0,\bI_p)$, and it approximates $N(0,\bI_p)$ only when $h$ is small. Therefore, when the target is the standard Gaussian, the discretized Langevin method only samples from an approximate target distribution.

The above calculation demonstrates that SFS performs better than the Euler-Maruyama discretization of the Langevin SDE in the canonical case when the target distribution is standard Gaussian. This suggests that SFS can be a more accurate and efficient sampler.

\section{Theoretical properties}\label{Theorey}
In this section, we establish the non-asymptotic bounds on the Wasserstein distance
between the law of the samples generated via SFS and the target distribution using Algorithms \ref{alg:1} or \ref{alg:2}. To this end, we further assume that the drift term $b(x,t)$ is Lipschitz continuous in $x$ and $\frac{1}{2}$-H{\"o}lder continuous in $t$, that is,
\begin{align}\label{cond3}
\|b(x,t)-b(y,s)\|_2\leq  C_1 \left(\|x-y\|_2+|t-s|^{\frac{1}{2}}\right), \
x, y\in \mathbb{R}^p \text{ and } t, s\in [0,1],
\tag{C3}
\end{align}
where $C_1>0$ is a finite constant. Obviously,  \eqref{cond3} implies  \eqref{cond2}  by setting $t=s$ in \eqref{cond3}.

\begin{remark}\label{R2}
Since $f(x) \varpropto g(x)$ defined in \eqref{gf},
 $g$ and $\nabla g$ are Lipschitz continuous and $g$ has a lower bound strictly greater than 0 imply \eqref{cond1} and \eqref{cond3},
 see Appendix \ref{append} for details.
\end{remark}

Let $\nu_1$ and $\nu_2$  be two probability measures defined on $\left(\mathbb{R}^{p},\mathcal{B}(\mathbb{R}^{p})\right)$, and denote $\mathcal{D}(\nu_1, \nu_2)$ as the collection of couplings $\nu$ on $\left(\mathbb{R}^{2p},\mathcal{B}(\mathbb{R}^{2p})\right)$ whose first and second marginal distributions are $\nu_1$ and $\nu_2$, respectively.  The Wasserstein distance of order $d \geq 1$ is defined as
\begin{align*}
W_{d}(\nu_1, \nu_2)=\inf _{\nu \in \mathcal{D}(\nu_1, \nu_2)}\left(\int_{\mathbb{R}^{p}} \int_{\mathbb{R}^{p}}\left\|\theta_1-\theta_2\right\|_2^{d} \mathrm{d} \nu\left(\theta_1,\theta_2\right)\right)^{1/d}.
\end{align*}

\subsection{Error bounds for SFS in Algorithm \ref{alg:1}}
 Let $Y_{t_K}$ be the value from the last iteration in Algorithm \ref{alg:1},  where
 we assume that the exact values of the drift term $b$ can be computed.

\begin{theorem}\label{th1}
Under the conditions \eqref{cond1} and \eqref{cond3}, we have
\begin{align}
\label{th1bound}
W_2(\mbox{Law} (Y_{t_K}),\mu)\leq
\mathcal{O}(\sqrt{ps}),
\end{align}
where $s=1/K$ is the step size.
\end{theorem}

\begin{remark}
The error bound in (\ref{th1bound}) is non-asymptotic in the sense that it holds for
any given values of the dimension $p$ and the step size $s$. The $O(1)$ factor in the bound only depends on the constants in the conditions \eqref{cond1} and \eqref{cond3}. This can be seen in the proof of Theorem \ref{th1} given in the appendix. Similar comments apply to Theorems \ref{th2}-\ref{th4} given below.
\end{remark}

\begin{remark}
For the Gaussian mixture distribution \eqref{gaumix}, \eqref{cond1} and \eqref{cond3}
 hold  if  $\bI_p>\Sigma_i, i=1,\ldots,\kappa$ and $f$ is bounded away from zero.
\end{remark}

\begin{remark}
The convergence rate $\sqrt{s}$ is the optimal strong convergence rate when using
the Euler-Maruyama discretization method for solving  SDE \citep{kloeden1992stochastic, weinan2019applied}.  The error rate only depends on the square root of the ambient dimension $p$,
but not on $p$ exponentially.    In the high-dimensional settings with large $p$, to ensure the error converges to zero, we can set the step size $s = o(1/p)$. In other words, he number of iterations $K=1/s$ of SFS only depends on $p$ super-linearly, but not exponentially.
Therefore,  SFS does not suffer from the curse of dimensionality.
\end{remark}

\subsection{Error bounds for SFS in Algorithm \ref{alg:2}}
Algorithm \ref{alg:2} deals with the case that when  the exact values of the drift term $b$ cannot be computed, only Monte Carlo approximations to $b$ are available.
To  establish the  non-asymptotic error bounds, we further assume that the potential $U(x,t)$ is strongly convex in $x$, i.e., there exists a finite  constant $M> 0$ such that for all $x,y \in \mathbb{R}^p$ and $t \in [0,1]$,
\begin{align}\label{cond4}
U(x,t)-U(y,t)-\nabla U(y,t)^{\top}(x-y)\geq (M/2)\left\|x-y\right\|^2_2. \tag{C4}
\end{align}
Without loss of generality, we can assume that $M <C_1$, where $C_1$ is given in \eqref{cond3}.
Condition \eqref{cond4}
is  similar to the condition H2 of \cite{chau2019stochastic} and the Assumption 3.2 of \cite{barkhagen2018stochastic}, which are used in their analysis of stochastic gradient
Langevin dynamics.
The strong convexity condition on the potential function $V$ is commonly assumed in the convergence analysis of Langevin algorithms
\citep{durmus2016high-dimensional,
durmus2016sampling,
durmus2017nonasymptotic,
dalalyan2017further,
dalalyan2017theoretical,
cheng2018convergence,
dalalyan2019user-friendly}.
These works focused on the Euler-Maruyama discretization of Langevin SDE \eqref{lange} and established non-asymptotic error bounds in Wasserstein distance, Kullback-Leibler divergence and total variation distance.

\begin{theorem}\label{th2}
Assume \eqref{cond1}, \eqref{cond3} and
\eqref{cond4} hold,  $g$ and $\nabla g$ are Lipschitz continuous, and $g$
is bounded below away from 0.
Then, for  $s<\frac{2}{C_1+M}<1$,
\begin{align*}
W_2(\mbox{Law}(\widetilde{Y}_{t_K}),\mu)\leq \mathcal{O}(\sqrt{ps})+\mathcal{O}\left(\sqrt{\frac{p}{\log(m)}}\right),
\end{align*}
 where $m$ is the size of the normal random sample used in approximating the drift term $b$ in (\ref{drifte1}) or (\ref{drifte2}).
\end{theorem}

\begin{remark}
This theorem provides some guidance on the selection of $s$ and $m$. To ensure convergence of the distribution of $\widetilde{Y}_{t_K}$, we should set the step size $s = o(1/p)$ and
$m=\exp(p/o(1))$.
In high-dimensional models with a large $p$, we need to generate a large number of random vectors from $N(0, \bI_p)$ to obtain an accurate estimate of the drift term $b$.
\end{remark}

If  we assume that $g$ is bounded above 
we can improve the non-asymptotic error bound, in which
$\mathcal{O}\left(\sqrt{p/\log(m)}\right)$
can be improved to be
$\mathcal{O}\left(\sqrt{p/m}\right)$.
\begin{theorem}\label{th3}
 Assume that, in addition to the conditions of Theorem  \ref{th2},
 $g$ is bounded above.
 Then, for  $s<\frac{2}{C_1+M}<1$,
\begin{align*}
W_2(\mbox{Law}(\widetilde{Y}_{t_K}),\mu)\leq \mathcal{O}(\sqrt{ps})+\mathcal{O}\left(\sqrt{\frac{p}{m}}\right).
\end{align*}
\end{theorem}

\begin{remark}
With the boundedness condition on $g$, to ensure convergence of the sampling distribution, we can set the step size $s = o(1/p)$ and
$m=p/o(1)$.
Note that the sample size requirement for approximating the drift term is significantly less
stringent than that in Theorem \ref{th2}.
\end{remark}

\subsection{Regularization to improve  the lower bound on $f$}
Theorem \ref{th1} is based on assumptions \eqref{cond1} and \eqref{cond3}, which hold
{if $g$ and $\nabla g$ are Lipschitz continuous and $g$ ($f$) has a lower bound strictly greater than 0.
}
Theorems \ref{th2}-\ref{th3}   also require that $g$ ($f$) is bounded away from zero. However, this  requirement  {\color{black}does} not hold  if the target distribution admits  compact support.
 To fix this pity,
we introduce an regularization on $\mu$ by mixing  $\mu$ and $N(0, \bI_p)$ together, i.e., considering
\begin{align}\label{muesp}
\mu_{\varepsilon}=(1-\varepsilon)\mu+\varepsilon N(0, \bI_p),
\end{align}
where $0<\varepsilon<1$.
Since the density ratio of $\mu_{\varepsilon}$ over $N(0, \bI_p)$ is
\begin{align}\label{fesp}
f_{\varepsilon}=\frac{d \mu_{\varepsilon}}{ dN(0,\bI_p)}=(1-\varepsilon)f+\varepsilon,
\end{align}
it is easy to deduce that $f_{\varepsilon}\geq \varepsilon>0$, and  $f_{\varepsilon}$ and $\nabla f_{\varepsilon}$ are Lipschitz continuous as long as $f$ and $\nabla f$ are, and $f_{\varepsilon}$ is bounded above if and only if  $f$ is bounded above.
 Obviously, $\mu_{\varepsilon}$  is a good approximation  to $\mu$ when $\varepsilon$ is small. We can sample from $\mu_{\varepsilon}$ using the the Euler-Maruyama disretization of SDE \eqref{sch-equation}
with the drift term $b(x, t)=\nabla \log Q_{1-t} f_{\varepsilon}(x)$.
For a given $\varepsilon \in (0, 1)$, we denote the samples generated using Algorithm \ref{alg:1} and Algorithm \ref{alg:2} with
$\mu_{\varepsilon}$ as the target distribution by
$\{Y_{t_k}(\varepsilon)\}_{k=0}^K$ and $\{\widetilde{Y}_{t_k}(\varepsilon)\}_{k=0}^K$, respectively.
We can prove the following  consistency  results for both  Algorithm \ref{alg:1} and Algorithm \ref{alg:2}.
\begin{theorem}\label{th4}
Under the conditions \eqref{cond1} and \eqref{cond3}, we have
\begin{align}
\label{th4a}
\underset{K \rightarrow \infty,\varepsilon \rightarrow 0}{\lim} W_2(\mbox{Law}(Y_{t_K}(\varepsilon)),\mu)= 0.
\end{align}
Assume \eqref{cond1}, \eqref{cond3} and
\eqref{cond4} hold, $g$ and that $\nabla g$ is Lipschitz continuous, we have
\begin{align}
\label{th4b}
\underset{m,K \rightarrow \infty, \varepsilon \rightarrow 0}{\lim} W_2(\mbox{Law}(\widetilde{Y}_{t_K}(\varepsilon)),\mu)=0.
\end{align}
\end{theorem}
\begin{remark}
As in Theorem \ref{th1}, the requirement on $s=1/K$ in both  (\ref{th4a}) and (\ref{th4b}) is $s=o(1/p)$. The requirement on $m$ is similar to that in Theorem \ref{th3}.
As can be seen in the proof of Theorem \ref{th4} in the appendix, to control the approximation error, we should take $\varepsilon = o(1/p)$.
\end{remark}

\section{Related work }
\label{related}
There is a large body of work on the MCMC
sampling algorithms based on the Langevin diffusion.
Convergence properties of the Langevin sampling algorithms have been established under three types of different assumptions: (a) the (strongly) convex potential assumption
\citep{durmus2016sampling,durmus2017nonasymptotic,dalalyan2017further,
dalalyan2017theoretical,cheng2018convergence,dalalyan2019user-friendly,
durmus2016high-dimensional}; (b)
the dissipativity condition for the drift term \citep{raginsky2017non,mou2019improved,zhang2019nonasymptotic}; (c)
  the local convexity condition for the potential function outside a ball \citep{durmus2017nonasymptotic,cheng2018sharp,ma2019sampling,bou2020coupling}.
However, these conditions may not hold for models with multiple modes, for example,   Gaussian mixtures, where their potentials are not convex and the Sobolev inequality may  not be satisfied.
Moreover, the constant in the log Sobolev inequality depends on the dimensionality exponentially \citep{menz2014poincare,wang2009log,hale2010asymptotic,raginsky2017non}, implying that the efficiency of Langevin samplers may suffer from the curse of dimensionality.
 SFS does not require the underlying Markov process to be ergodic, therefore, our results in Theorem \ref{th1}  do not  need the  conditions used in the analysis of Langevin samplers, and the error bounds only depend on the square root of the ambient dimension.
In particular, these two convergence results are applicable to Gaussian mixtures, in which the drift term
of the Schr\"{o}dinger-F\"{o}llmer diffusion can be calculated analytically. However, in Theorems \ref{th2} to \ref{th3}, where only Monte Carlo approximation to the drift  are available,  we need the strongly  convexity condition for the potential of the Schr\"{o}dinger-F\"{o}llmer diffusion. We believe that the strongly convexity condition on $U(x,t)$ with respect to  $x$ is not essential for SFS and we assume this condition for just technical purpose in the proofs.


The Schr\"{o}edinger bridge has been shown to have close connections with a number of problems in statistical physics,  optimal transport  and optimal control \citep{leonard2014survey}. Hoverer,  only a few recent works  use this tool for statistical sampling.
A Schr\"{o}edinger bridge sampler was recently proposed in \citep{bernton2019schrodinger}.
 For a given distribution $\mu$,
 the authors propose to iteratively modify the transition kernels of the reference Markov chain to obtain a process whose marginal distribution at the terminal time is approximately $\mu$.
A second recent work considered the Schr\"{o}dinger bridge problem when only samples from the initial and the target distributions are available \citep{pavon2018datadriven}. The authors proposed an iterative
procedure that uses constrained maximum likelihood estimation and importance sampling to estimate the functions solving the Schr\"{o}dinger system.
The algorithms developed in these papers are inspired by the iterative proportional fitting  procedure, or the Sinkhorn algorithm \citep{sinkhorn1964, peyre2020computational}.
\citet{tzen2019theoretical,sbg21,de2021diffusion} considered the problem of learning a generative model from samples based on the Schr\"{o}dinger-F\"{o}llmer diffusion with an unknown drift term and estimate the drift via deep neural networks.
The problem settings and the tools used in the aforementioned works are different from
the present work.
We use Schr\"{o}dinger-F\"{o}llmer diffusion  (\ref{sde}) as a  sampler for unnormalized distributions.

\section{Numerical studies}\label{simulation}
We conduct numerical experiments to evaluate the effectiveness of SFS.
We consider sampling from some one-dimensional, two-dimensional Gaussian mixture distributions and Bayesian Logistic regression.
The R code and the Python code of SFS are available at
\url{https://github.com/Liao-Xu/SFS_R}
and \url{https://github.com/Liao-Xu/SFS_py}, respectively.

We compare SFS with the MH algorithm  \citep{metropolis1953equation,hastings1970monte,chib1995understanding,robert1999metropolis},
Hamiltonian Monte Carlo (HMC) \citep{duane1987hybrid,neal2011mcmc},
Stochastic Gradient Hamiltonian Monte Carlo (SGHMC) \citep{chen2014stochastic}, Unadjusted Langevin Algorithm (ULA) \citep{durmus2016high-dimensional,durmus2016sampling,dalalyan2017theoretical, durmus2017nonasymptotic,cheng2018convergence}, Stochastic Gradient Langevin Dynamics (SGLD) \citep{welling2011bayesian,ahn2012bayesian,patterson2013stochastic},
cyclical Stochastic Gradient Langevin Dynamics (cSGLD) \citep{zhang2019cyclical}, No U-Turn Sampler (NUTS) \citep{hoffman2014no, betancourt2017conceptual}, and Haario Bardenet Adaptive Metropolis MCMC (ACMC) \citep{johnstone2016uncertainty, haario2001adaptive}.

 In our experiments,
 we use the R package \texttt{mcmc} \citep{brooks2011handbook, tierney1994markov} for the MH algorithm and the R packages \texttt{sde} \citep{iacus2009simulation} and \texttt{yuima} \citep{iacus2018simulation}  for the ULA.
Moreover, we use the code from \citep{chen2014stochastic} for  HMC and SGHMC and the code from \citep{zhang2019cyclical} for SGLD and cSGLD in our numerical studies.
We use the Python library PINTS \citep{Clerx2019Pints} for implementing NUTS and ACMC.
\subsection{One-dimensional Gaussian mixture distribution}
In this section, we consider three one-dimensional Gaussian mixture distributions,
\begin{align}\label{onegaussion1}
f_1(x)&=0.5 N(x;-2,0.5^2)+0.5 N(x;2,0.5^2), \\
f_2(x)&=0.5N(x;-4,0.5^2)+0.5 N(x;4,0.5^2),
\label{onegaussion2} \\
f_3(x)&=0.5 N(x;-8,0.5^2)+0.5N(x;8,0.5^2). \label{onegaussion3}
\end{align}
The variances for each component of these three distributions are fixed, but the centroids are further apart from models \eqref{onegaussion1} to \eqref{onegaussion3}.
Therefore, sampling becomes more difficult from models \eqref{onegaussion1} to \eqref{onegaussion3}.
We use the proposed SFS and other methods to generate samples for these three mixture distributions.
We set the sample size $N$ to be 5,000 and the grid hyperparameter
$K$ to be 100 in Algorithm \ref{alg:1}.

In Fig. \ref{fig:1d},
for all models \eqref{onegaussion1} to \eqref{onegaussion3}, we show the curves of kernel density estimation from all methods using different colors and line types while the target density functions are shaded in grey.
When the centroids of Gaussians are close, the proposed SFS, MH and SGHMC perform comparably well but samples from other methods collapse on one mode as shown in Fig. \ref{fig:1d} (a). In the case that the centroids of Gaussians move apart from each other, only samples from SFS can accurately represent the underlying target distribution while all other methods collapse on one mode, see Fig. \ref{fig:1d} (b) and (c).
\begin{figure}[H] 
	\centering
	 \begin{minipage}[b]{0.32\textwidth}
        \centering
         \includegraphics[width=\textwidth]{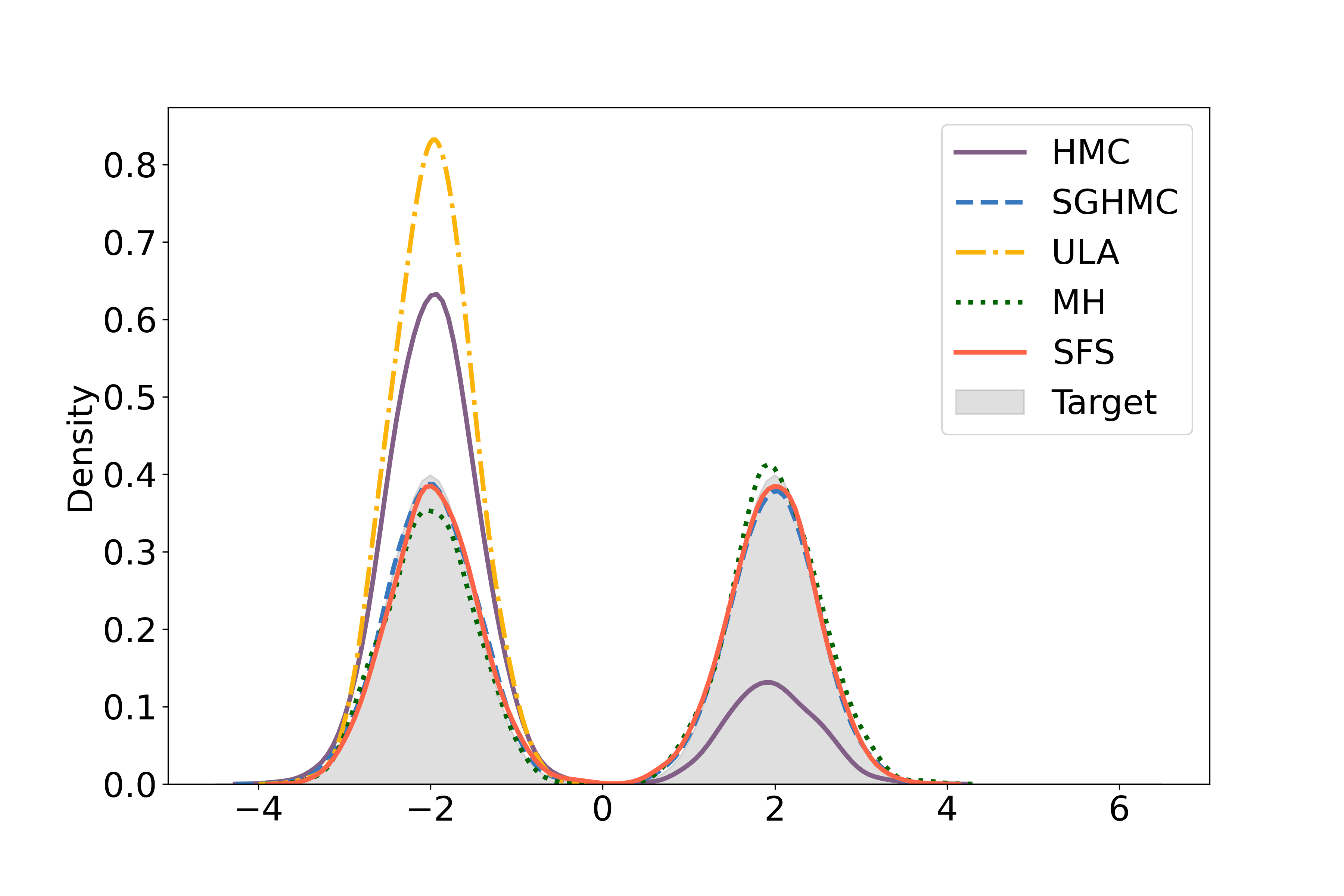}\\
         (a) $f_1(x)$
     	\end{minipage}
	  \begin{minipage}[b]{0.32\textwidth}
        \centering
         \includegraphics[width=\textwidth]{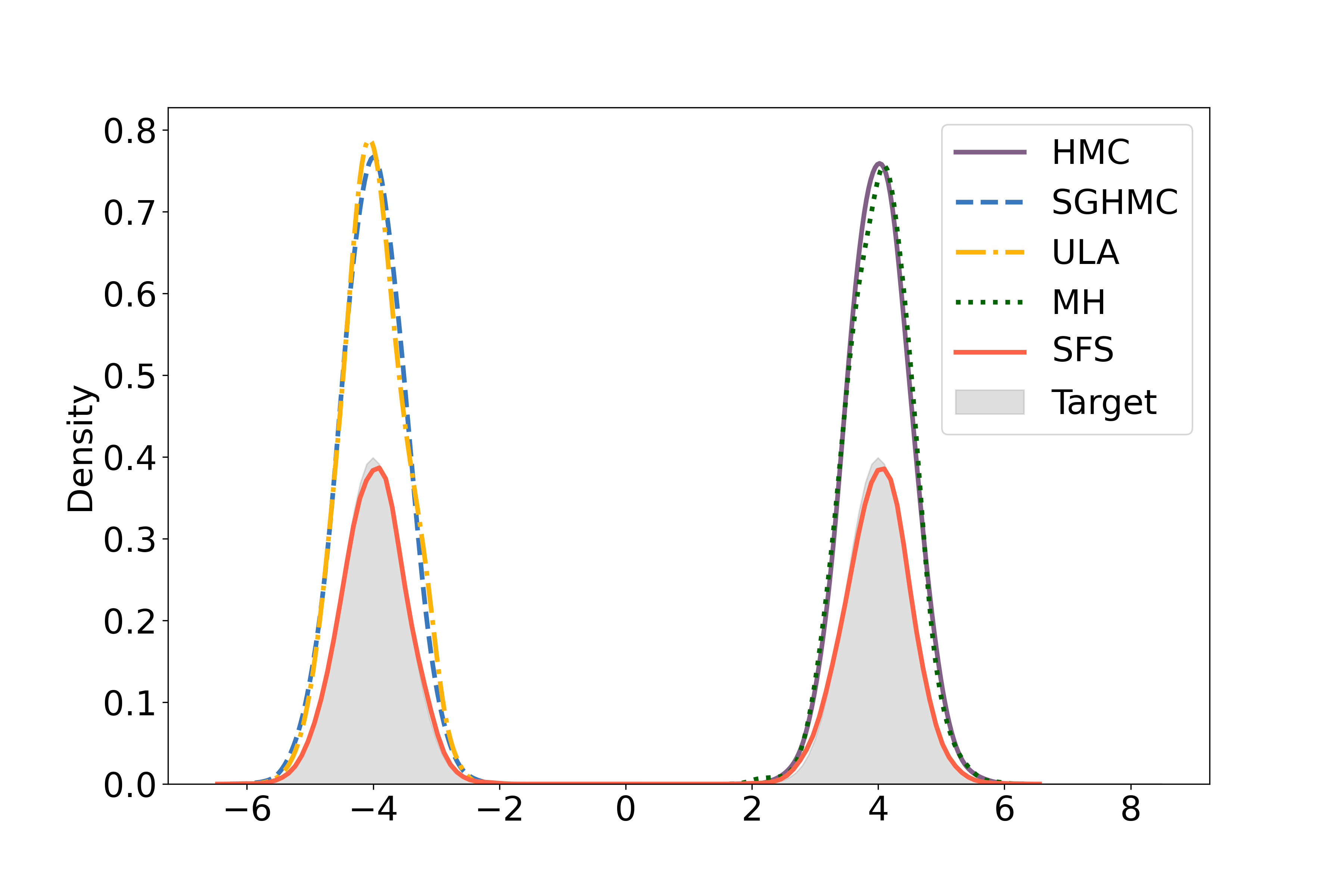}\\
         (b) $f_2(x)$
     	\end{minipage}
	 \begin{minipage}[b]{0.32\textwidth}
        \centering
         \includegraphics[width=\textwidth]{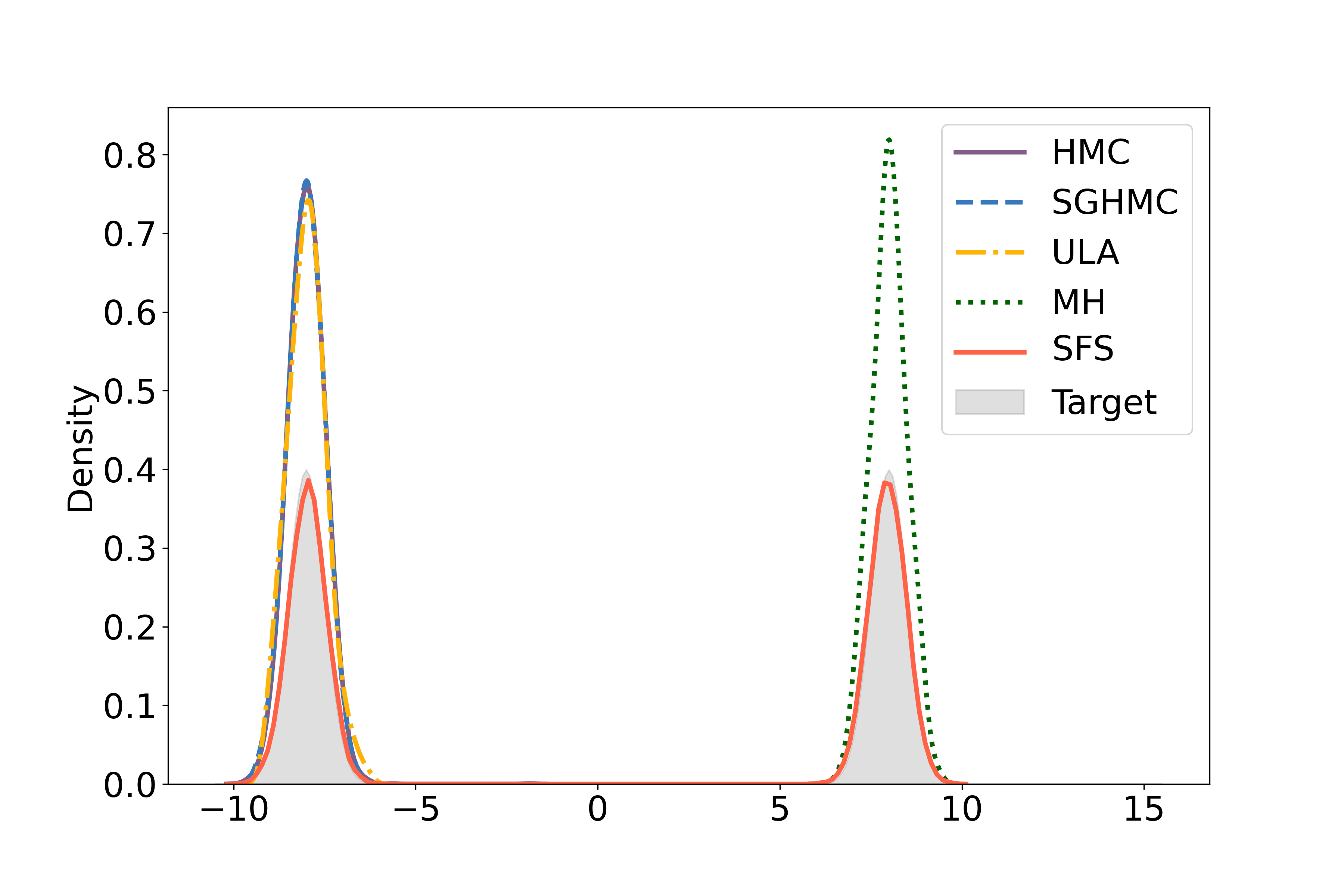}\\
        (c) $f_3(x)$
     	\end{minipage}

	\caption{Kernel density estimation of one-dimensional Gaussian mixture distributions. The target distribution is shaded in grey.}
	\label{fig:1d}
\end{figure}
\subsection{Two-dimensional Gaussian mixture distribution}
In this section, we consider four sets of two-dimensional Gaussian  mixture distributions,
\begin{align}\label{gaussian_circle}
\kappa&=4, 8, 16,~
\alpha_i=\lambda_1 (\sin (2(i-1) \pi / \kappa), \cos (2(i-1) \pi / \kappa)),~
i = 1,\cdots, \kappa,~
\lambda_1 = 2,4,8,\\
\kappa&=16,~
\balpha=(\lambda_2 \{-3,-1,1,3\})^{\top} \times(\lambda_2 \{-3,-1,1,3\}),~
\lambda_2 = 1,1.5,2, \label{gaussian16}\\
\kappa&=25,~
\balpha=(\lambda_3 \{-2,-1,0,1,2\})^{\top} \times(\lambda_3 \{-2,-1,0,1,2\}),~
\lambda_3 = 2,3,  \label{gaussian25}\\
\kappa&=49,~
\balpha=(\lambda_4 \{-3,-2,-1,0,1,2,3\})^{\top} \times(\lambda_4 \{-3,-2,-1,0,1,2,3\}),~
\lambda_4 = 2,3. \label{gaussian49}
\end{align}
We set the proportions $\theta_i=1/\kappa$
and the covariance matrices $\Sigma_i= 0.03\cdot \bI_2$ across these settings.
The centroids of Gaussian components gradually become farther apart across these scenarios for each setting.
We set $N=20,000$, $K=100$ for model \eqref{gaussian_circle} and $K=200$ for the rest of the settings in Algorithm \ref{alg:1}.
The centroids of the Gaussian components in model \eqref{gaussian_circle} form a circle while those in \eqref{gaussian16}-\eqref{gaussian49} form a square matrix.
For all models from \eqref{gaussian_circle} to \eqref{gaussian49}, we employ the proposed SFS and other methods to generate samples and then visualize the kernel density estimation in Fig. \ref{fig:circle}-\ref{fig:49}, respectively.

As shown in Fig. \ref{fig:circle}-\ref{fig:49} for models  \eqref{gaussian_circle}-\eqref{gaussian49}, only the samples from SFS succeed in estimating the underlying target distribution while all other methods collapse on one or a few modes when the models becomes more difficult to sample from.
In model \eqref{gaussian_circle}, although samples from MH, SGLD, SGHMC can give a density estimation that matches the underlying target in the simple case with four centroids in the first row of Fig. \ref{fig:circle}, they all fail to accurately sample from the target distribution when the numbers of components are bigger.

When the centroids from the mixture components form a square matrix shape, only the MH algorithm can perform as well as SFS in a simpler case (the first row of Fig. \ref{fig:16} to \ref{fig:49}), all the other methods collapse on one or few modes. When the models become more difficult to sample from, mode collapsing for other methods becomes even worse. In all the simulations, we observe that only samples generated via SFS can accurately produce a density estimation that matches the underlying mixture distribution.

In general, we can draw the conclusion that SFS outperforms other algorithms, including MH, ULA, SGLD, SGHMC, cSGLD, NUTS and ACMC, through the visualization of two-dimensional Gaussian sampling.

\begin{figure}[H]
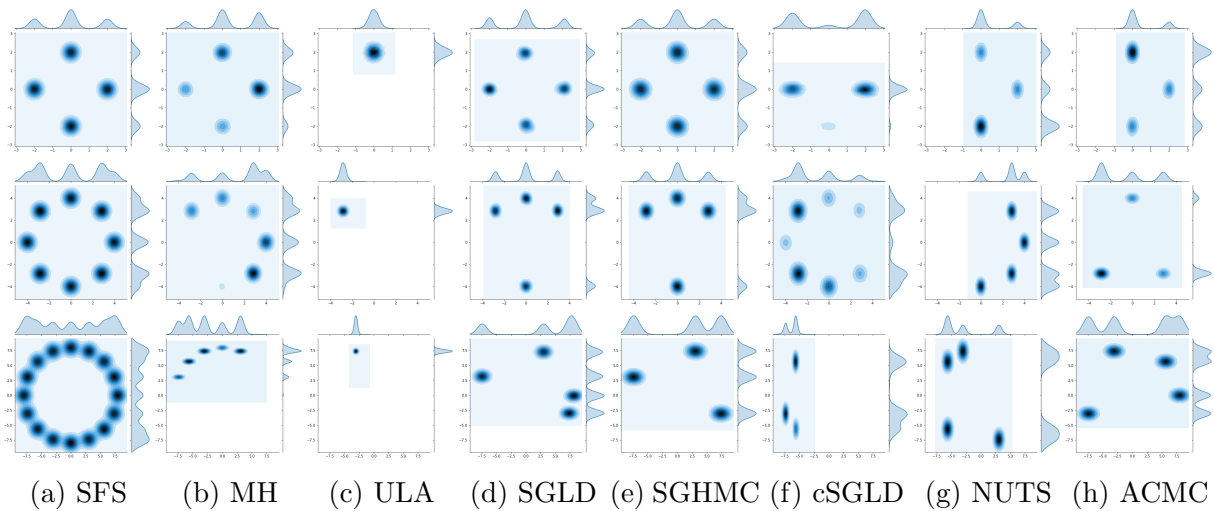

   \foreach \x in { SFS, MH,ULA,SGLD,SGHMC,cSGLD,NUTS, ACMC}{
     \begin{subfigure}[b]{0.125\textwidth}
        \centering
         \includegraphics[width=\textwidth]{Figures/\x _kappa_4_RA_margin}
     \end{subfigure}
     \hspace*{-0.9em}
     }
\\
      \foreach \x in { SFS, MH,ULA,SGLD,SGHMC,cSGLD,NUTS, ACMC}{
     \begin{subfigure}[b]{0.125\textwidth}
         \centering
         \includegraphics[width=\textwidth]{Figures/\x _kappa_8_RB_margin}
     \end{subfigure}
     \hspace*{-0.9em}
     }
    \\
    \foreach \x in { SFS, MH,ULA,SGLD,SGHMC,cSGLD,NUTS, ACMC}{
     \begin{subfigure}[b]{0.125\textwidth}
         \centering
         \includegraphics[width=\textwidth]{Figures/\x _kappa_16_RC_margin}
          \caption{\x}
     \end{subfigure}
     \hspace*{-0.9em}
     }
           \caption{Kernel density estimation with marginal distribution plots for the circle shaped Gaussian mixture distributions. The number of Gaussian components $\kappa=4, 8, 16$ and the radius of the circle gradually increases from the top row to the bottom row.}
        \label{fig:circle}
\end{figure}
\begin{figure}[H]
   \foreach \x in { SFS, MH,ULA,SGLD,SGHMC,cSGLD,NUTS, ACMC}{
     \begin{subfigure}[b]{0.125\textwidth}
        \centering
         \includegraphics[width=\textwidth]{Figures/\x _kappa_16_A_margin}
     \end{subfigure}
     \hspace*{-0.9em}
     }
\\
      \foreach \x in { SFS, MH,ULA,SGLD,SGHMC,cSGLD,NUTS, ACMC}{
     \begin{subfigure}[b]{0.125\textwidth}
         \centering
         \includegraphics[width=\textwidth]{Figures/\x _kappa_16_B_margin}
     \end{subfigure}
     \hspace*{-0.9em}
     }
    \\
    \foreach \x in { SFS, MH,ULA,SGLD,SGHMC,cSGLD,NUTS, ACMC}{
     \begin{subfigure}[b]{0.125\textwidth}
         \centering
         \includegraphics[width=\textwidth]{Figures/\x _kappa_16_C_margin}
          \caption{\x}
     \end{subfigure}
     \hspace*{-0.9em}
     }
\caption{Kernel density estimation with marginal distribution plots for the 16-mode Gaussian mixture distributions. The distance between the neighboring modes gradually increases from the top row to the bottom row.}
        \label{fig:16}
\end{figure}
\begin{figure}[H]
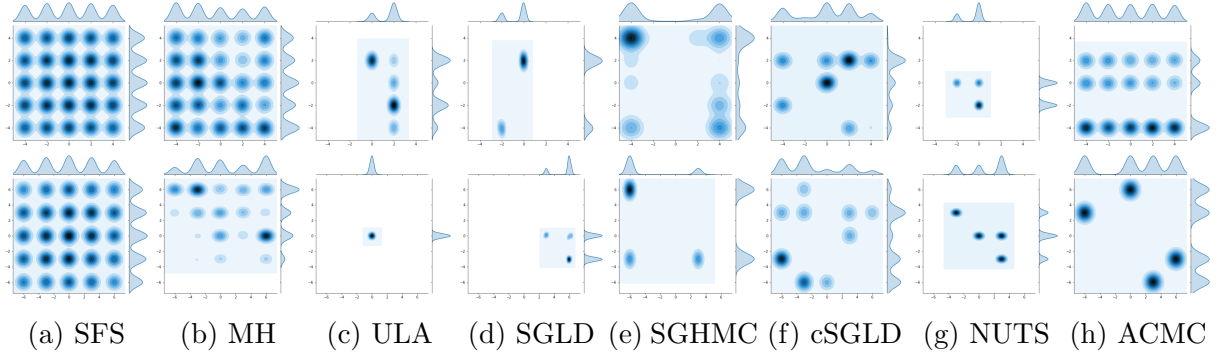

   \foreach \x in { SFS, MH,ULA,SGLD,SGHMC,cSGLD,NUTS, ACMC}{
     \begin{subfigure}[b]{0.125\textwidth}
        \centering
         \includegraphics[width=\textwidth]{Figures/\x _kappa_25_A_margin}
     \end{subfigure}
     \hspace*{-0.9em}
     }
\\
      \foreach \x in { SFS, MH,ULA,SGLD,SGHMC,cSGLD,NUTS, ACMC}{
     \begin{subfigure}[b]{0.125\textwidth}
         \centering
         \includegraphics[width=\textwidth]{Figures/\x _kappa_25_B_margin}
          \caption{\x}
     \end{subfigure}
     \hspace*{-0.9em}
     }
           \caption{Kernel density estimation with marginal distribution plots for the 25-mode Gaussian mixture distributions. From the first row to the second row. The distance between the neighboring modes gradually increases.}
        \label{fig:25}
\end{figure}
\begin{figure}[H]
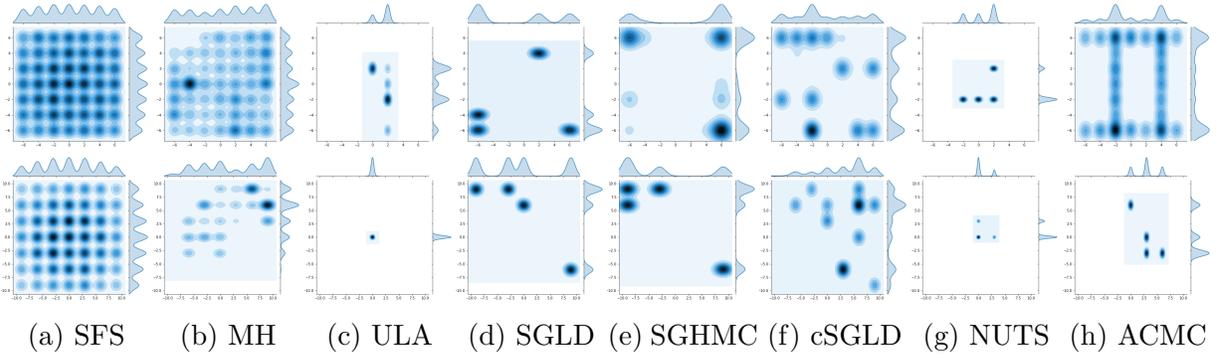

     \foreach \x in { SFS, MH,ULA,SGLD,SGHMC,cSGLD,NUTS, ACMC}{
     \begin{subfigure}[b]{0.125\textwidth}
        \centering
         \includegraphics[width=\textwidth]{Figures/\x _kappa_49_A_margin}
     \end{subfigure}
     \hspace*{-0.9em}
     }
\\
       \foreach \x in { SFS, MH,ULA,SGLD,SGHMC,cSGLD,NUTS, ACMC}{
     \begin{subfigure}[b]{0.125\textwidth}
         \centering
         \includegraphics[width=\textwidth]{Figures/\x _kappa_49_B_margin}
          \caption{\x}
     \end{subfigure}
     \hspace*{-0.9em}
     }
           \caption{Kernel density estimation with marginal distribution plots for 49-mode Gaussian mixture distributions simulation. The distance between the neighboring modes gradually increases from the top row to the bottom row.}
        \label{fig:49}
\end{figure}
\subsection{Evaluation of the balance of modes}
We apply the $k$-means algorithm \citep{macqueen1967some, forgy1965cluster, lloyd1982least} to estimate the proportions of the clusters using samples generated using SFS and the other methods consider here.
In Fig. \ref{fig:circle_mean}, we use dotted red lines for the true proportions in the target mixture distribution and color bars for the estimated proportions of the samples.
Clearly, only the proportions of the components estimated based on the samples from SFS can accurately match the proportions of the components in the target mixture distribution of model \eqref{gaussian_circle}, suggesting that SFS performs balanced sampling.

 \begin{figure}[hbt]
 \begin{center}
     \begin{minipage}[hbt]{0.6\textwidth}
        \centering
         \includegraphics[width=\textwidth]{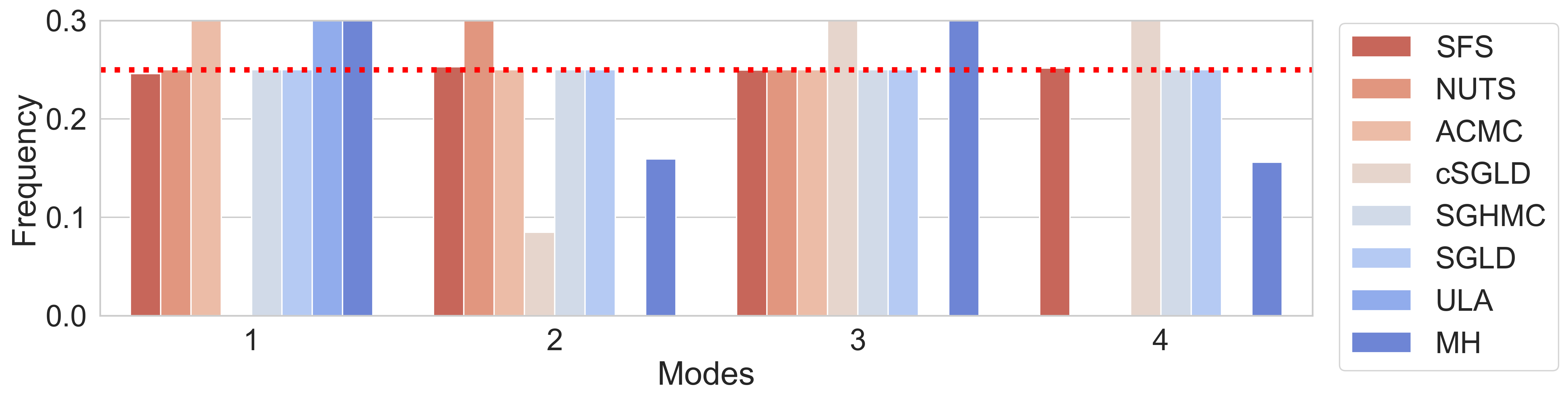}\\
         Scenario 1, $\lambda_1=2$
     \end{minipage}
          \begin{minipage}[hbt]{0.6\textwidth}
        \centering
         \includegraphics[width=\textwidth]{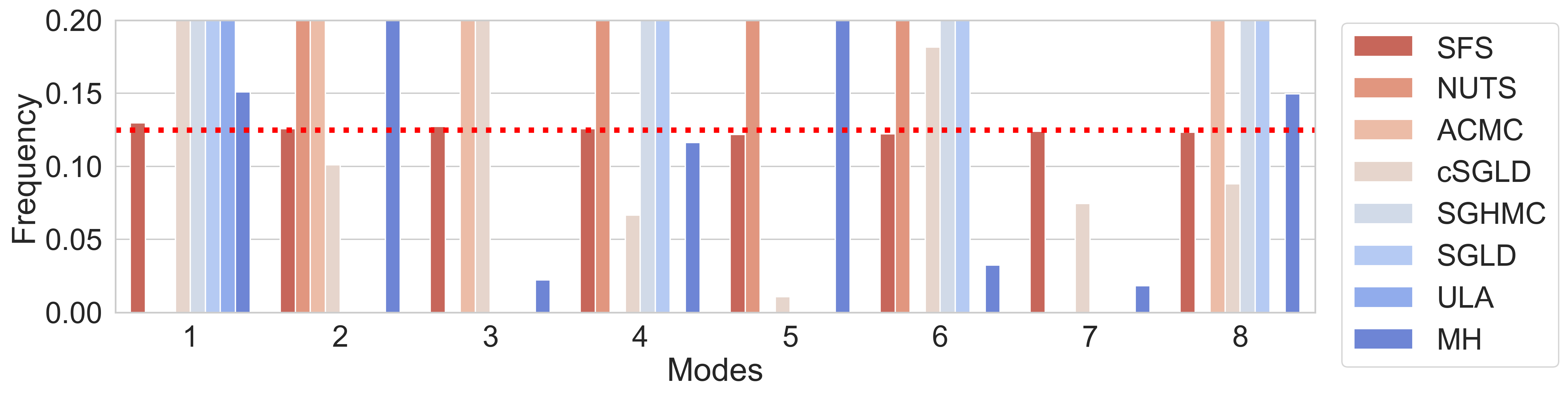}\\
         Scenario 2, $\lambda_1=4$
     \end{minipage}
          \begin{minipage}[hbt]{0.6\textwidth}
        \centering
         \includegraphics[width=\textwidth]{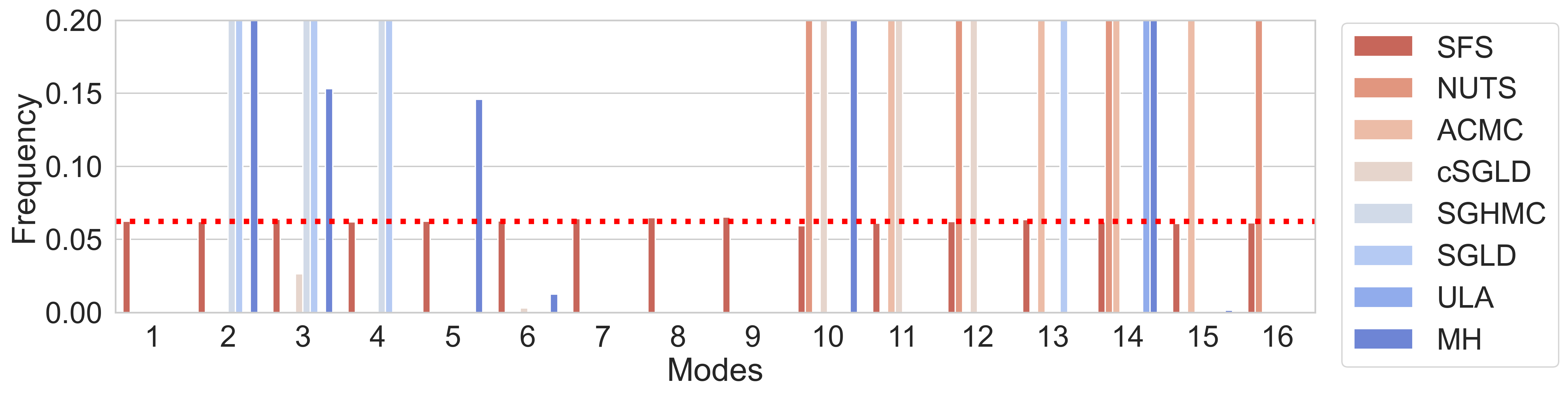}\\
         Scenario 3, $\lambda_1=8$
     \end{minipage}
           \caption{Frequency labeled by $k$-means on the circle shaped Gaussian mixture distributions simulation.}
        \label{fig:circle_mean}
 \end{center}
\end{figure}
\subsection{Bayesian logistic regression}
Consider  the binary logistic regression with an independent and identically distributed sample $\{x_i,y_i\}_{i=1}^{n}$, that is,
\begin{align*}
P(y_i=1|x_i)=\frac{\exp(x^{\top}_i\beta^*)}{1+\exp(x^{\top}_i\beta^*)},~i=1,\ldots,n,
\end{align*}
where  $x_i \in \mathbb{R}^p$  is  the covariate vector, $y_i\in \{0,1\}$ is the response variable and $\beta^*=(\beta^*_1,\ldots,\beta_p^*)^{\top}\in\mathbb{R}^p$  is the vector of underlying regression coefficients.
We draw samples for $\beta^*$ from its posterior distribution.
Following \citep{durmus2016high-dimensional}, \citep{durmus2016sampling} and \citep{dalalyan2017theoretical}, we set the prior distribution of $\beta^*$ to be a Gaussian distribution with zero mean and covariance matrix $\Sigma_{\beta^*}=(\sum_{i=1}^n x_ix_i^{\top}/n)^{-1}$, then the posterior distribution of $\beta^*$ is expressed as
\begin{align}
\label{Blogit}
\mu(\beta^*)\propto \exp\left(\sum_{i=1}^n \left(y _ix_i^{\top}\beta^*-\log(1+\exp(x_i^{\top}\beta^*))\right)-{\beta^*}^{\top}\Sigma_{\beta^*}^{-1}\beta^*/2\right).
\end{align}
Also, as in \citep{durmus2016high-dimensional}, \citep{durmus2016sampling} and \citep{dalalyan2017theoretical},  we  set $n=1000$ and $p=5$.
The covarites vector $\{x_i\}_{i=1}^n$ are independent and identically distributed from $N(0,\Sigma)$, where $\Sigma_{i,j}=0.5^{|i-j|}$ for $1\leq i,j\leq p$.
In Algorithm \ref{alg:2}, we set $K=200$ and $m=1000$.
We generate a random sample with sample size $N=10,000$ using the proposed SFS,  MH algorithm and ULA.
We compare SFS with MH algorithm and ULA in term of
the sample mean (Mean), sample median (Med) and sample variance (Var).
The results are reported in Table \ref{tab:1} and Fig. \ref{bayes_box} and \ref{bayes_density} below.

In Table \ref{tab:1}, the first column describes the methods; the second column shows the criterions
including Mean, Med and Var; the
values in brackets are the true value of  underlying regression coefficients  $\beta^*_i$ with $i=1,\ldots,5$.
For the proposed method,
the values of  Mean and Med are close to the true values of $\beta^*$ except $\beta_2^*$,  while it still takes more accurate values  than MH algorithm and ULA in  $\beta_2^*$,
and the values of Var are slightly larger than other two methods.
Especially,  MH algorithm takes the values trifle away from
$\beta^*_1$, $\beta^*_2$, $\beta^*_4$ and $\beta^*_5$  in Mean and Med.

Fig. \ref{bayes_box} shows the box plots of samples for SFS, MH algorithm and ULA, where the horizontal dotted red lines are the ground truth of the coefficients. We can see that the samples of SFS match the ground truth better than MH algorithm and ULA.

Fig. \ref{bayes_density} depicts the curves of kernel density estimation and true values for each component of $\beta^*$.  In this figure, the red solid lines and the corresponding red shades are the kernel density estimation of SFS, while the blue and green lines correspond to the MH algorithm and ULA respectively, and the
red dotted lines perpendicular to the horizontal axis represent the value of $\beta_i^*$ with $i=1,\dots,5$.
Except for $\beta^*_2$,  the rest $\beta^*_i$ coincide with the peaks of kernel density estimation curves of the proposed method.
But, $\beta_2^*$ and  $\beta_3^*$  keep away from  the peaks of the kernel density estimation curves of ULA, and each component of $\beta^*$ is not close to that of MH algorithm.
  \begin{table}
  \caption{\label{tab:1}Numerical results of the binary Bayesian logistic regression.}
\begin{tabular}{*{7}{c}}
\toprule
Method & Criterions & $\beta^*_1$ (0.220) & $\beta^*_2$ (0.208) & $\beta^*_3$ (-2.027) & $\beta^*_4$ (0.744) & $\beta^*_5$ (1.424) \\ \hline
SFS    & Mean       & 0.206               & 0.307               & -2.033               & 0.716               & 1.372               \\
       & Med        & 0.205               & 0.316               & -2.068               & 0.724               & 1.390               \\
       & Var        & 0.041               & 0.049               & 0.098                & 0.049               & 0.057               \\ \hline
MH     & Mean       & 0.279               & 0.415               & -2.286               & 0.862               & 1.560               \\
       & Med        & 0.336               & 0.469               & -2.321               & 0.965               & 1.556               \\
       & Var        & 0.010               & 0.018               & 0.015                & 0.029               & 0.003               \\ \hline
ULA    & Mean       & 0.219               & 0.372               & -2.214               & 0.780               & 1.461               \\
       & Med        & 0.220               & 0.371               & -2.213               & 0.776               & 1.463               \\
       & Var        & 0.011               & 0.013               & 0.027                & 0.016               & 0.014               \\ \bottomrule

\end{tabular}
\end{table}
\begin{figure}[H]
\begin{center}
   \foreach \x in { 1,2,3,4,5}{
     \begin{subfigure}[b]{0.25\textwidth}
        \centering
         \includegraphics[width=\textwidth]{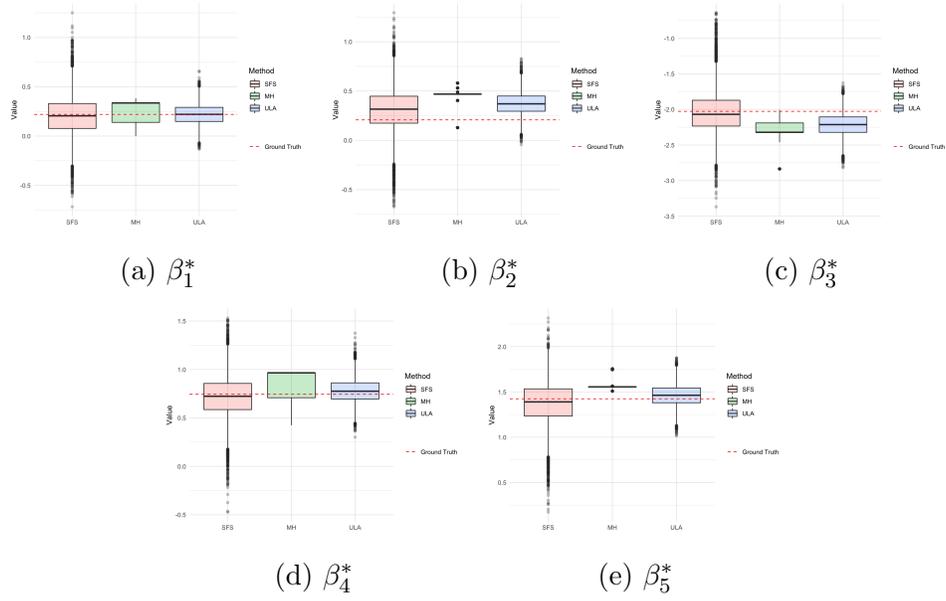}
        \caption{$\beta_\x^*$}
     \end{subfigure}
     }

      \caption{Box plots of the binary Bayesian logistic regression with $N=10,000$ for $\beta_1^*,  \dots, \beta_5^*$. The horizontal dotted red lines are the ground truth of the coefficients.}
        \label{bayes_box}
             \end{center}
\end{figure}

\begin{figure}[H]
\begin{center}
   \foreach \x in { 1,2,3,4,5}{
     \begin{subfigure}[b]{0.25\textwidth}
        \centering
         \includegraphics[width=\textwidth]{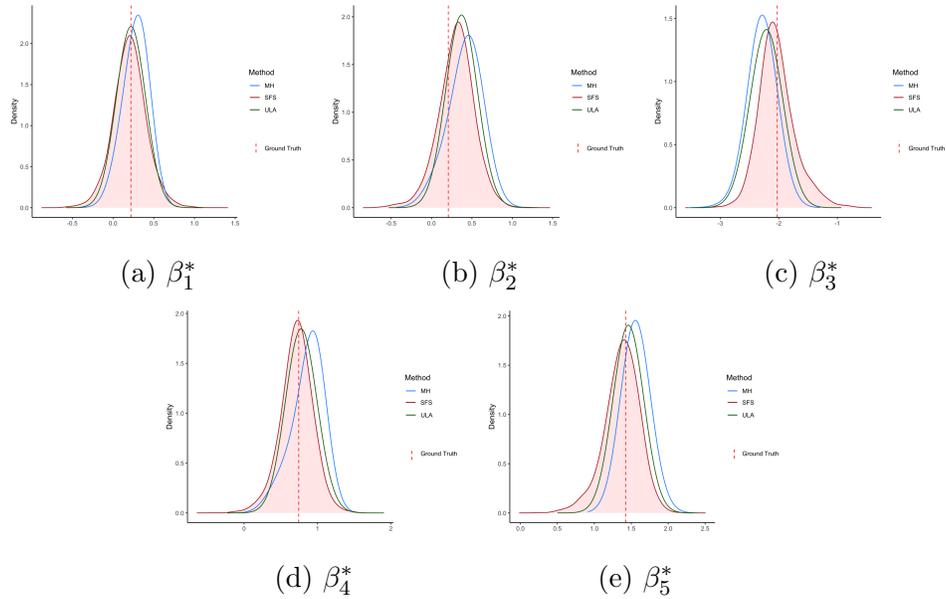}
         \caption{$\beta_\x^*$}
     \end{subfigure}
     }
      \caption{Kernel density estimation of the binary Bayesian logistic regression with $N=10,000$ for $\beta_1^*,  \dots, \beta_5^*$. The red solid lines and the corresponding red shades are the estimated density curves of the proposed method, and the blue and green lines correspond to the MH algorithm and ULA respectively. The vertical dotted red lines are the ground truth of the coefficients.}
        \label{bayes_density}
         \end{center}
\end{figure}

To demonstrate the unbiasedness and coverage rates, we use the same experimental setting for the binary Bayesian logistic regression, except that $p=3$. We generate 100 replicates of  random samples with sample size $N=1,000$ using the proposed SFS,  MH algorithm and ULA.

The absolute errors and ordinary errors are respectively reported with the box plots in the Figure \ref{abs_box} and \ref{box}.
As Figure \ref{abs_box} shows, the performance of SFS and other algorithms is comparable in terms of estimation accuracy.
We can conclude from the box plots of errors in Figure \ref{box} that the estimation of SFS is unbiased based on the experimental simulation.
We demonstrate the coverage rate of the 95\% confidence interval for the true coefficient in Table \ref{suptab:1}, which is based on the random samples we generate from the posterior distribution.
We can see that the coverage rate of the confidence interval based on SFS approaches 95\%, which is generally comparable to that of ULA and MH algorithm.

We apply SFS to the real data classification tasks to compare ULA and MH algorithm.
The dataset Diabetes \citep{smith1988using} aims to predict whether a person has diabetes according to health measurements.
The objective of the dataset German \citep{Dua:2019} is to classify good or bad credit risks.
We need to classify the survival status in the dataset Haberman \citep{haberman1976generalized, Dua:2019} based on the descriptions of patients.
The aim of the dataset Vc \citep{berthonnaud2005analysis, Dua:2019} is to predict orthopaedic patients into normal or abnormal.
The sample size, attribute number and URL of these data are shown in Table \ref{suptab:2}.
We can conclude from the prediction accuracy shown in Figure \ref{real} that SFS outperforms the other two methods.

\section{Conclusion}\label{conclusion}
We propose SFS method for sampling from  distributions
via the Euler-Maruyama discretization of Schr{\"o}dinger-F{\"o}llmer diffusion defined on the unite time interval $[0,1]$.
{ A prominent feature of SFS  is that it does not require the underlying Markov process to be ergodic to generate samples from the target distribution while ergodicity is needed   for other MCMC samplers.}
We establish non-asymptotic error bounds for the sampling distribution of the SFS in the Wasserstein distance under appropriate conditions.
{In particular, when the drift term can be evaluated analytically, we only need some smoothness conditions on the target distribution to ensure the error bounds for samples generated in Algorithm \ref{alg:1}. For the drift term without analytical expression, we further propose to generate samples in Algorithm \ref{alg:2} using the Monte Carlo estimation and establish the consistency of Algorithm \ref{alg:2}.}
Our numerical experiments demonstrate that SFS can enlarge the scope of existing samplers for sampling in multi-mode distributions possibly without normalization.
Therefore, the proposed SFS can be a useful addition to the existing methods for sampling from possibly unnormalized distributions.

Several problems deserve further study. For example, for the samples generated from Algorithm \ref{alg:2}, we show its convergence under a strong convexity condition on the potential. It would be interesting to weaken or even remove this condition. This is an interesting and challenging technical problem.
Applying  SFS to more Bayesian inference with more complex real data  is also of immense interest.

\begin{table}
\begin{center}
  \caption{\label{suptab:1}Coverage rates of the binary Bayesian logistic regression (100 replicates).}
\begin{tabular}{@{}ccccc@{}}
\toprule
Components & 1  & 2   & 3  \\ \midrule
SFS        & 99 \% & 98 \% & 99 \%\\
MH         & 93 \%& 95  \%& 96 \%\\
ULA        & 98 \%& 98 \%& 100 \%\\ \bottomrule
\end{tabular}
\end{center}
\end{table}
\begin{figure}[H]
\begin{center}
   \foreach \x in { 1,2,3}{
     \begin{subfigure}[b]{0.3\textwidth}
        \centering
         \includegraphics[width=\textwidth]{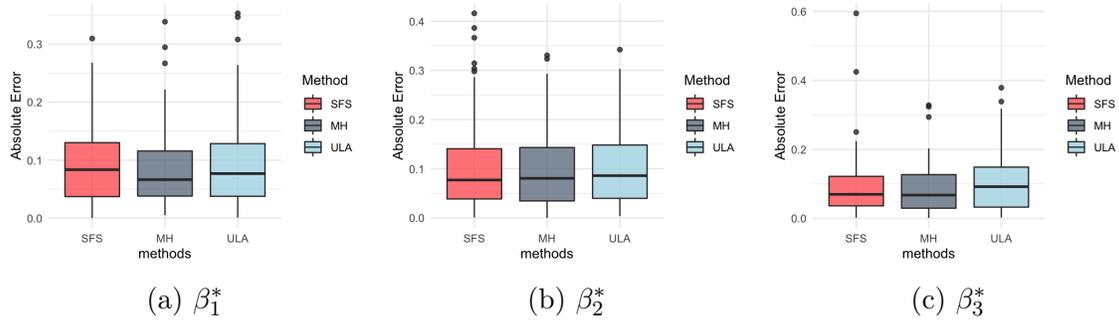}
        \caption{$\beta_\x^*$}
     \end{subfigure}
     }
      \caption{Box plots of absolute errors on the binary Bayesian logistic regression with $N=1,000$ for $\beta_1^*,  \dots, \beta_3^*$.}
        \label{abs_box}
             \end{center}
\end{figure}

\begin{figure}[H]
\begin{center}
   \foreach \x in { 1,2,3}{
     \begin{subfigure}[b]{0.3\textwidth}
        \centering
         \includegraphics[width=\textwidth]{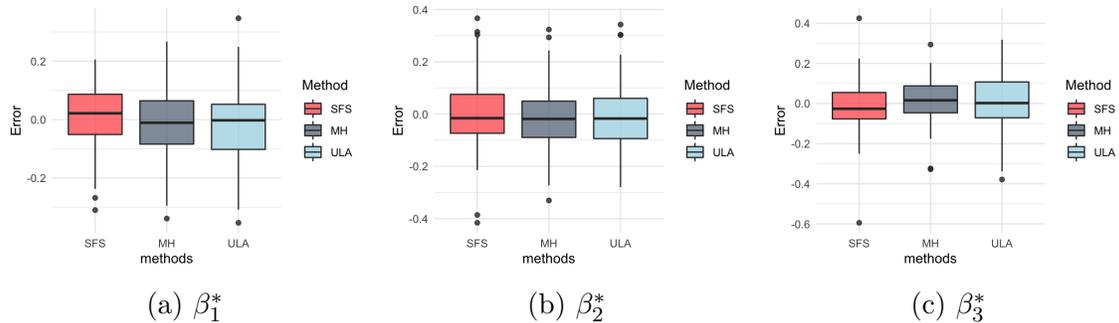}
        \caption{$\beta_\x^*$}
     \end{subfigure}
     }

      \caption{Box plots of errors on the binary Bayesian logistic regression with $N=1,000$ for $\beta_1^*,  \dots, \beta_3^*$.}
        \label{box}
             \end{center}
\end{figure}

  \begin{table}
  \caption{\label{suptab:2}Descriptions on real datasets.}
  \begin{center}
\tiny
\begin{tabular}{ccccl}
\toprule
Dataset  &\# Sample & \# Attributes & URL                                                                                           \\ \midrule
Diabetes & 768         & 8                    & \url{https://www.kaggle.com/uciml/pima-indians-diabetes-database}           \\
German   & 1000        & 20                   & \url{https://archive.ics.uci.edu/ml/ datasets/statlog+(german+credit+data)} \\
Haberman & 306         & 3                    & \url{https://archive.ics.uci.edu/ml/datasets/haberman's+survival}           \\
Vc       & 310         & 6                    & \url{https://archive.ics.uci.edu/ml/datasets/Vertebral+Column}              \\ \bottomrule
\end{tabular}
\end{center}
\end{table}

\begin{figure}[H]
\begin{center}
   \foreach \x in {Diabetes, German,  Haberman,  Vc}{
     \begin{subfigure}[b]{0.4\textwidth}
        \centering
         \includegraphics[width=\textwidth]{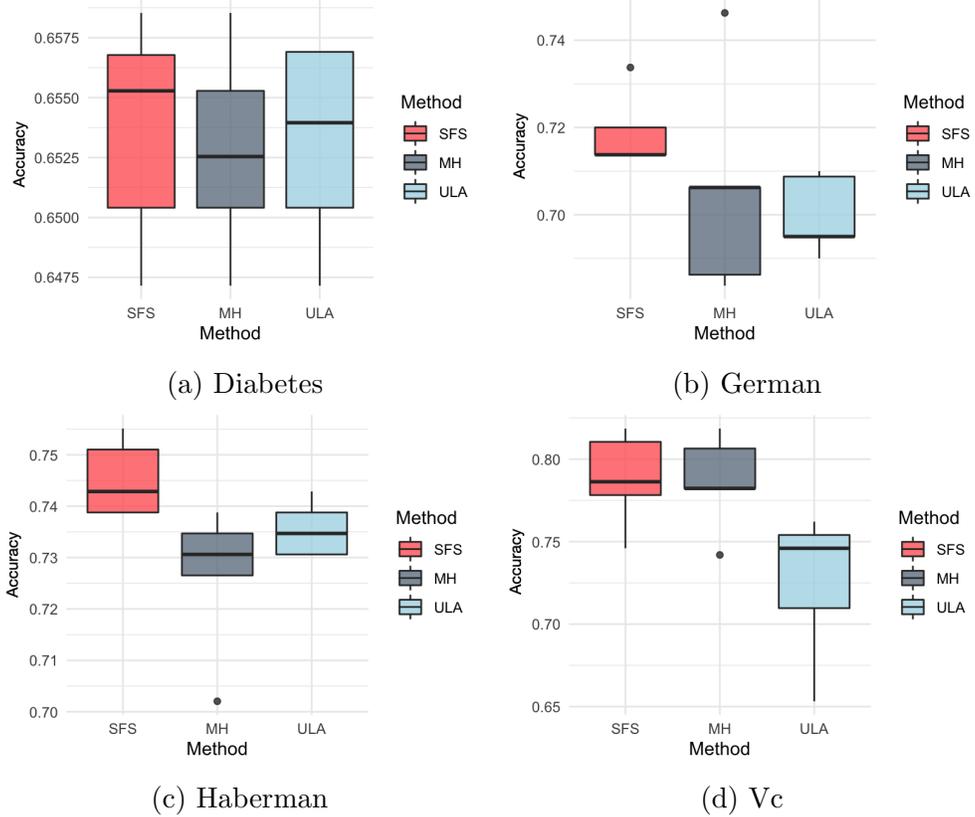}
         \caption{\x}
     \end{subfigure}
     }
      \caption{Prediction accuracies on real data with SFS, ULA and MH algorithm; for each dataset, we evaluate these methods by selecting training and validation sets randomly for five times.}
        \label{real}
         \end{center}
\end{figure}



\section{Acknowledgments}
The authors would like to thank Dr. Francisco Vargas at Department of Computer Science, Cambridge University for helpful  discussion.

J. Huang is partially supported by the U.S. NSF grant DMS-1916199. Y. Jiao is supported in part
by the National Science Foundation of China under Grant 11871474 and by the research fund of
KLATASDSMOE of China. J. Liu is supported by the grants MOE2018-T2-1-046 and MOE2018-T2-2-006 from the Ministry of Eduction, Singapore.  The numerical studies of this work were done on the supercomputing system in the
Supercomputing Center of Wuhan University.

\begin{appendix}
\section{Appendix}\label{append}
\setcounter{equation}{0}
\def\theequation{A.\arabic{equation}}

In the appendix, we prove
Proposition \ref{SBP}, Remark \ref{R2} and  Theorems \ref{th1}-\ref{th4}.
We note that Proposition \ref{SBP} is a known result
(see, e.g., \cite{follmer1985, follmer1986,tzen2019theoretical}).
We include a proof here for ease of reference.
\subsection{Proof of Proposition \ref{SBP}}
\begin{proof}
It is well-known that
the transition probability density of a standard $p$-dimensional Brownian motion is given by
\begin{align*}
\widetilde{p}_{s, t}(x, y)=\frac{1}{(2 \pi(t-s))^{p / 2}} \exp \left(-\frac{1}{2(t-s)}\|x-y\|^{2}_2\right).
\end{align*}
This implies that the diffusion process $\{X_t\}_{t\in [0,1]}$ defined in \eqref{SBP} admits the transition probability density
\begin{align*}
p_{s, t}(x, y)=\widetilde{p}_{s, t}(x, y) \frac{Q_{1-t} f(y)}{Q_{1-s} f(x)}.
\end{align*}
It follows that for any measurable set $A \in \mathcal{B}(\mathbb{R}^p)$,
\begin{align*}
P(X_1\in A)&=\int_A p_{0,1}(0,y) \mathrm{d} y\\
&=\int_A \widetilde{p}_{0,1}(0,y)\frac{Q_{0} f(y)}{Q_{1} f(0)}\mathrm{d} y\\
&=\mu(A).
\end{align*}
Therefore, $X_1$ is distributed as the probability distribution $\mu$.
This completes the proof.
\end{proof}
{\color{black}
\subsection{Drift term $b(x,t)$ for Gaussian mixture distribution \eqref{gaumix}}
We derive the expression of the drift term $b(x,t)$ in \eqref{gaumix} for the Gaussian mixture distribution.
\begin{proof} For $Z \sim N(0, \bI_p)$, we have
\begin{align}\label{exp21}
&\Ebb_{Z} f_i(x+\sqrt{1-t}Z)\notag\\
&=\Ebb_{Y\sim N(x,(1-t)\bI_p)} f_i(Y)\notag\\
&=\frac{1}{|\Sigma_i|^{1/2}}\cdot\Ebb_{Y\sim N(x,(1-t)\bI_p)} \exp\left(\frac{\|Y\|_2^2-(Y-\alpha_i)^{\top}\Sigma_i^{-1}(Y-\alpha_i)}{2}\right)\notag\\
&=\frac{1}{(2\pi (1-t))^{p/2}|\Sigma_i|^{1/2}} \int \exp\left(\frac{\|y\|^2}{2}-\frac{(y-\alpha_i)^{\top}\Sigma_i^{-1}(y-\alpha_i)}{2}-\frac{\|y-x\|^2_2}{2-2t}\right)d y\notag\\
&=\frac{g_i(x,t)}{|t\Sigma_i+(1-t)\bI_p|^{1/2}}.
\end{align}
where
\begin{align*}
g_i(x,t)&=\exp\left(\frac{1}{2-2t}\|(t\bI_p+(1-t)\Sigma_i^{-1})^{-1/2}((1-t)\Sigma_i^{-1}\alpha_i+x)\|_2^2\right)\\
&~~~~\times \exp\left(-\frac{1}{2}\alpha_i^T\Sigma_i^{-1}\alpha_i-\frac{1}{2-2t}\|x\|_2^2
\right).
\end{align*}
Similarly, we have
\begin{align}\label{exp22}
&\Ebb_{Z}\nabla f_i(x+\sqrt{1-t}Z)\notag\\
&=\Ebb_{Y\sim N(x,(1-t)\bI_p)}\nabla f_i(Y)\notag\\
&=\frac{\Sigma_i^{-1}\alpha_i}{|\Sigma_i|^{1/2}}\cdot\Ebb_{Y\sim N(x,(1-t)\bI_p)} \exp\left(\frac{\|Y\|_2^2-(Y-\alpha_i)^{\top}\Sigma_i^{-1}(Y-\alpha_i)}{2}\right)\notag\\
&~~~+\frac{\bI_p-\Sigma_i^{-1}}{|\Sigma_i|^{1/2}}\cdot\Ebb_{Y\sim N(x,(1-t)\bI_p)} Y\cdot\exp\left(\frac{\|Y\|_2^2-(Y-\alpha_i)^{\top}\Sigma_i^{-1}(Y-\alpha_i)}{2}\right)
\notag\\
&=
\frac{\Sigma_i^{-1}\alpha_i}{|\Sigma_i|^{1/2}}\cdot
\int \frac{1}{(2\pi (1-t))^{p/2}} \exp\left(\frac{\|y\|^2}{2}-\frac{(y-\alpha_i)^{\top}\Sigma_i^{-1}(y-\alpha_i)}{2}-\frac{\|y-x\|^2_2}{2-2t}\right)d y\notag\\
&~~~+\frac{\bI_p-\Sigma_i^{-1}}{|\Sigma_i|^{1/2}}\cdot
\int \frac{y }{(2\pi (1-t))^{p/2}} \exp\left(\frac{\|y\|^2}{2}-\frac{(y-\alpha_i)^{\top}\Sigma_i^{-1}(y-\alpha_i)}{2}-\frac{\|y-x\|^2_2}{2-2t}\right)d y\notag\\
&=\frac{\Sigma_i^{-1}\alpha_i+(\bI_p-\Sigma_i^{-1})[t\bI_p+(1-t)\Sigma_i^{-1}]^{-1}
[(1-t)\Sigma_i^{-1}\alpha_i+x]}
{|t\Sigma_i+(1-t)\bI_p|^{1/2}}\cdot g_i(x,t).
\end{align}
Hence, the analytical expression of $b(x,t)$ in \eqref{dr} is obtained by plugging the expressions \eqref{exp21}
and \eqref{exp22} into \eqref{dr}.
\end{proof}
}
\subsection{Proof of Remark \ref{R2}}
\begin{proof}
Since $g$ and $\nabla g$ are Lipschitz continuous,  there exists a finite and positive constant $\gamma$
such that for all $x,y \in \mathbb{R}^p$,
\begin{align}\label{L1L}
|g(x)-g(y)|\leq \gamma \|x-y\|_2,
\end{align}
\begin{align}\label{L2L}
\|\nabla g(x)- \nabla g(y)\|_2\leq \gamma \|x-y\|_2.
\end{align}
Moreover, since  $g$ has a lower  bound greater than 0,   there exists a finite and positive constant $\xi$ such that
\begin{align}\label{boundfL}
 g \geq \xi>0.
\end{align}
By \eqref{L1L} and \eqref{boundfL}, it yields that
for all $x \in \mathbb{R}^{p}$ and $t \in[0,1]$,
\begin{align}\label{re3}
\|b(x, t)\|_2=\frac{\left\|\nabla Q_{1-t} g(x)\right\|_2}{Q_{1-t} g(x)} \leq \frac{\gamma}{\xi}.
\end{align}
Then, by \eqref{L1L}-\eqref{re3},  for all  $x, y \in \mathbb{R}^{p}$ and $t \in[0,1]$,
\begin{align*}
\left\|b(x, t)-b\left(y, t\right)\right\|_2 &=\left\|\frac{\nabla Q_{1-t} g(x)}{Q_{1-t} g(x)}-\frac{\nabla Q_{1-t} g\left(y\right)}{Q_{1-t} g\left(y\right)}\right\|_2 \\
& \leq \frac{\left\|\nabla Q_{1-t} g(x)-\nabla Q_{1-t} g\left(y\right)\right\|_2}{Q_{1-t} g\left(y\right)}+\|b(x, t)\|_2 \cdot \frac{\left|Q_{1-t} g(x)-Q_{1-t} g\left(y\right)\right|}{Q_{1-t} g\left(y\right)} \\
&\leq\left(\frac{\gamma}{\xi}+\frac{\gamma^{2}}{\xi^{2}}\right)\left\|x-y\right\|_2.
\end{align*}
Similarly, by \eqref{L1L}-\eqref{re3},  for all  $x, y \in \mathbb{R}^{p}$ and $t \in[0,1]$,
\begin{align*}
\left\|b(x, t)-b\left(x, s\right)\right\|_2 &=\left\|\frac{\nabla Q_{1-t} g(x)}{Q_{1-t} g(x)}-\frac{\nabla Q_{1-s} g\left(x\right)}{Q_{1-s} g\left(x\right)}\right\|_2 \\
& \leq \frac{\left\|\nabla Q_{1-t} g(x)-\nabla Q_{1-s} g\left(x\right)\right\|_2}{Q_{1-t} g\left(x\right)}+\|b(x, s)\|_2 \cdot \frac{\left|Q_{1-s} g(x)-Q_{1-t} g\left(x\right)\right|}{Q_{1-t} g\left(x\right)} \\
&\leq\Ebb\|Z\|_2\left(\frac{\gamma}{\xi}+\frac{\gamma^{2}}{\xi^{2}}\right)|t-s|^{\frac{1}{2}}\\
&\leq\sqrt{p}\left(\frac{\gamma}{\xi}+\frac{\gamma^{2}}{\xi^{2}}\right)|t-s|^{\frac{1}{2}}.
\end{align*}
Therefore, there exists a finite and positive constant $C_1\geq \sqrt{p}\left(\frac{\gamma}{\xi}+\frac{\gamma^{2}}{\xi^{2}}\right)$ such that
\begin{align*}
\|b(x,t)-b(y,s)\|_2&\leq\|b(x,t)-b(x,s)\|_2+\|b(x,s)-b(y,s)\|_2\\
&\leq C_1\left(\|x-y\|_2+|t-s|^{\frac{1}{2}}\right).
\end{align*}
Therefore, the assumption  \eqref{cond3} holds.
Further, set $t=s$, then  \eqref{cond2} holds.
Combining  \eqref{re3} and \eqref{cond2} with the triangle inequality, we have
$$
\|b(x,t)\|_2\leq \|b(0,t)\|_2+C_1\|x\|_2\leq \frac{\gamma}{\xi}+C_1\|x\|_2.
$$
Let $C_0\geq \max\left\{\frac{\gamma}{\xi}, C_1\right\}$, then
\eqref{cond1} holds.

\end{proof}
\subsection{Preliminary lemmas for Theorem \ref{th1}}
First, we introduce  Lemmas \ref{lemma1o}-\ref{lemma2o} in preparing for the
proofs of Theorem \ref{th1}.
\begin{lemma}\label{lemma1o}
Assume \eqref{cond1} holds, then
\begin{align*}
\Ebb[\|X_t\|_2^2]\leq 2(C_0+p)\exp(2C_0t).
\end{align*}
\end{lemma}
\begin{proof}
By the definition of $X_t$ in \eqref{sch-equation}, we have
$
\|X_t\|_2\leq \int_{0}^t\|b(X_u,u)\|_2\mathrm{d}u+\|B_t\|_2.
$
It follows that
\begin{align*}
\|X_t\|_2^2&\leq
2\left(\int_{0}^t\|b(X_u,u)\|_2\mathrm{d}u\right)^2+2\|B_t\|_2^2\\
&\leq
2t\int_{0}^t\|b(X_u,u)\|_2^2\mathrm{d}u+2\|B_t\|_2^2\\
&\leq
2t\int_{0}^tC_0[\|X_u\|_2^2+1]\mathrm{d}u+2\|B_t\|_2^2,
\end{align*}
where the first inequality follows from the inequality $(a+b)^2\leq 2a^2+2b^2$, the last inequality follows from condition \eqref{cond1}. Thus,
\begin{align*}
\Ebb\|X_t\|_2^2&\leq
2t\int_{0}^tC_0(\Ebb\|X_u\|_2^2+1)\mathrm{d}u+2\Ebb\|B_t\|_2^2\\
&\leq
2C_0\int_{0}^t \Ebb\|X_u\|_2^2\mathrm{d}u+2(C_0+p).
\end{align*}
By the Bellman-Gronwall inequality, we have
\begin{align*}
\Ebb\|X_t\|_2^2\leq 2(C_0+p)\exp(2C_0t).
\end{align*}
This completes the proof.
\end{proof}
\begin{lemma}\label{lemma2o}
Assume \eqref{cond1} holds, then for any $0\leq t_1\leq t_2\leq 1$,
\begin{align*}
\Ebb[\|X_{t_2}-X_{t_1}\|_2^2]\leq 4C_0\exp(2C_0)(C_0+p)(t_2-t_1)^2+2C_0(t_2-t_1)^2+2p(t_2-t_1).
\end{align*}
\end{lemma}
\begin{proof}
By the definition of $X_t$ in \eqref{sch-equation}, we have
\begin{align*}
\|X_{t_2}-X_{t_1}\|_2\leq \int_{t_1}^{t_2}\|b(X_u,u)\|_2\mathrm{d}u+\|B_{t_2}-B_{t_1}\|_2.
\end{align*}
Therefore,
\begin{align*}
\|X_{t_2}-X_{t_1}\|_2^2&\leq
2\left(\int_{t_1}^{t_2}\|b(X_u,u)\|_2\mathrm{d}u\right)^2+2\|B_{t_2}-B_{t_1}\|_2^2\\
&\leq
2(t_2-t_1)\int_{t_1}^{t_2}\|b(X_u,u)\|_2^2\mathrm{d}u+2\|B_{t_2}-B_{t_1}\|_2^2\\
&\leq
2(t_2-t_1)\int_{t_1}^{t_2}C_0[\|X_u\|_2^2+1]\mathrm{d}u+2\|B_{t_2}-B_{t_1}\|_2^2,
\end{align*}
where the last inequality follows from condition \eqref{cond1}.
Hence,
\begin{align*}
\Ebb\|X_{t_2}-X_{t_1}\|_2^2&\leq
2(t_2-t_1)\int_{t_1}^{t_2}C_0(\Ebb\|X_u\|_2^2+1)\mathrm{d}u
+2\Ebb\|B_{t_2}-B_{t_1}\|_2^2\\
&\leq
4C_0\exp(2C_0)(C_0+p)(t_2-t_1)^2+2C_0(t_2-t_1)^2+2p(t_2-t_1),
\end{align*}
where the last inequality follows from Lemma \ref{lemma1o}.
This completes the proof.
\end{proof}
\subsection{Proof of Theorem \ref{th1}}
\begin{proof}
By the definition of $Y_{t_k}$ and $X_{t_k}$, we have
\begin{align*}
\|Y_{t_k}-X_{t_k}\|_2^2
&\leq \|Y_{t_{k-1}}-X_{t_{k-1}}\|_2^2
+\left(\int_{t_{k-1}}^{t_k}\|b(X_u,u)-b(Y_{t_{k-1}},t_{k-1})\|_2\mathrm{d} u\right)^2\\
&~~~+2\|Y_{t_{k-1}}-X_{t_{k-1}}\|_2
\left(\int_{t_{k-1}}^{t_k}\|b(X_u,u)-b(Y_{t_{k-1}},t_{k-1})\|_2\mathrm{d} u\right)\\
&\leq (1+s) \|Y_{t_{k-1}}-X_{t_{k-1}}\|_2^2
+(1+s)\int_{t_{k-1}}^{t_k}\|b(X_u,u)-b(Y_{t_{k-1}},t_{k-1})\|_2^2\mathrm{d} u\\
&\leq (1+s) \|Y_{t_{k-1}}-X_{t_{k-1}}\|_2^2+4C_1^2(1+s)\int_{t_{k-1}}^{t_k}[\|X_u-Y_{t_{k-1}}\|_2^2+|u-t_{k-1}|]\mathrm{d} u\\
&\leq (1+s) \|Y_{t_{k-1}}-X_{t_{k-1}}\|_2^2
+8C_1^2(1+s)\int_{t_{k-1}}^{t_k}\|X_u-X_{t_{k-1}}\|_2^2\mathrm{d} u\\
&~~~~+8C_1^2s(1+s)\|X_{t_{k-1}}-Y_{t_{k-1}}\|_2^2
+4C_2^2(1+s)s^2\\
&\leq (1+s+8C_1^2(s+s^2))\|Y_{t_{k-1}}-X_{t_{k-1}}\|_2^2
+8C_1^2(1+s)\int_{t_{k-1}}^{t_k}\|X_u-X_{t_{k-1}}\|_2^2\mathrm{d} u\\
&~~~~+4C_1^2(1+s)s^2,
\end{align*}
where the second inequality holds duo to $2ab\leq s a^2+\frac{b^2}{s}$,
the third inequality follows from condition \eqref{cond3}.
Then,
\begin{align}\label{th1eq1o}
\Ebb\|Y_{t_k}-X_{t_k}\|_2^2
&\leq (1+s+8C_1^2(s+s^2))
\Ebb\|Y_{t_{k-1}}-X_{t_{k-1}}\|_2^2\notag\\
&~~~+8C_1^2(1+s)\int_{t_{k-1}}^{t_k}\Ebb\|X_u-X_{t_{k-1}}\|_2^2\mathrm{d} u+4C_1^2(s^2+s^3)\notag\\
&\leq(1+s+8C_1^2(s+s^2)) \Ebb\|Y_{t_{k-1}}-X_{t_{k-1}}\|_2^2+h(s)+4C_1^2(s^2+s^3),
\end{align}
where $h(s)=8C_0^2(s+s^2)[4C_0(C_0+p)\exp(2C_0)s^2+2C_0s^2+2ps]$,
and the last inequality \eqref{th1eq1o} follows from Lemma \ref{lemma2o}.
Owing to $Y_{t_0}=X_{t_0}=0$, we can conclude that
\begin{align*}
\Ebb\|Y_{t_K}-X_{t_K}\|_2^2
&\leq\frac{(1+s+8C_1^2(s+s^2))^K-1}{s+8C_1^2(s+s^2)}
\left[h(s)+4C_1^2(s^2+s^3)\right]\\
&\leq\mathcal{O}(ps).
\end{align*}
Therefore,
\begin{align*}
W_2(Law(Y_{t_K}),\mu)
\leq \mathcal{O}(\sqrt{ps}).
\end{align*}
This completes the proof.
\end{proof}
\subsection{Preliminary Lemmas for Theorem \ref{th2}-\ref{th3}}
First, we introduce  Lemmas \ref{lemma3}-\ref{lemma6} in preparing for the
proofs of Theorem \ref{th2}-\ref{th3}.
\begin{lemma}\label{lemma3}
(Lemma 1 in \cite{dalalyan2017further} and Lemma 2 in \cite{dalalyan2019user-friendly}).
Denote
$\Delta_k=X_{t_k}-\widetilde{Y}_{t_k}$ and
$\omega_k=b(\widetilde{Y}_{t_k},t_k)-b(X_{t_k},t_k)$ with $k=0,1,\ldots,K$.
Assume conditions \eqref{cond2} and \eqref{cond4} hold, and $s<2/(C_1+M)<1$, then
\begin{align*}
\|\Delta_{k}-s\omega_k\|_2\leq \rho \|\Delta_{k}\|_2,
\end{align*}
where $\rho=1-sM\in (0,1)$.
\end{lemma}
\begin{proof}
By  \eqref{cond2} and \eqref{cond4}, we have
$$
(x-y)^{\top}(\nabla U(x,t)-\nabla U(y,t)) \geq \frac{MC_1}{M+C_1}\|x-y\|_{2}^{2}+\frac{1}{M+C_1}\|\nabla U(x,t)- \nabla U(y,t)\|_{2}^{2}
$$
for all $x, y \in \mathbb{R}^p$. Therefore, by some simple calculation, we can get
\begin{align*}
\left\|\Delta_{k}-s \omega_{k}\right\|_{2}^{2} &=\left\|\Delta_{k}\right\|_{2}^{2}-2 s \Delta_{k}^{\top} \omega_{k}+s^{2}\left\|\omega_{k}\right\|_{2}^{2} \\
&\leq\left\|\Delta_{k}\right\|_{2}^{2}-\frac{2 s M C_1}{M+C_1}\left\|\Delta_{k}\right\|_{2}^{2}-\frac{2 s}{M+C_1}\left\|\omega_{k}\right\|_{2}^{2}+s^{2}\left\|\omega_{k}\right\|_{2}^{2} \\
&=\left(1-\frac{2 s M C_1}{M+C_1}\right)\left\|\Delta_{k}\right\|_{2}^{2}
+s\left(s-\frac{2}{M+C_1}\right)\left\|\omega_{k}\right\|_{2}^{2}.
\end{align*}
By \eqref{cond4}, we have $\left\|\omega_{k}\right\|_{2}\geq M\|\Delta_{k}\|_2$. Due to $s \leq 2/(M+C_1)$, then it yields
\[
\left\|\Delta_{k}-s \omega_{k}\right\|_{2}^{2} \leq(1-s M)^{2}\left\|\Delta_{k}\right\|_{2}^{2}.
\]
As $C_1> M$, then $2/(M+C_1)< 1/M$. Thus $\rho=1-sM\in (0,1)$.
This completes the proof.
\end{proof}
\begin{lemma}\label{lemma4}
If $g$  and $\nabla g$ are Lipschitz continuous, and $g$  has the lower bound greater than 0, then for any $R>0$,
\begin{align*}
\underset{\|x\|_2
\leq R,t\in [0,1]}{\sup} \Ebb\left[\|b(x,t)-\tilde{b}_m(x,t)\|_2^2\right]
\leq 
\mathcal{O}\left(\frac{p\exp(R^2)}{m}\right).
\end{align*}
Moreover, if $g$ has a finite upper bound, then
\begin{align*}
\underset{x\in\mathbb{R}^p,t\in [0,1]}{\sup}
\Ebb\left[\|b(x,t)-\tilde{b}_m(t,x)\|_2^2\right]
\leq \mathcal{O}\left(\frac{p}{m}\right).
\end{align*}
\end{lemma}
\begin{proof}
Denote two independent sets of independent copies of $Z\sim N(0, \bI_p)$,
that is, $\bZ=\{Z_1,\ldots,Z_m\}$ and $\bZ^{'}=\{Z_1^{'},\ldots,Z_m^{'}\}$.
For notation convenience, we  denote
\begin{align*}
&d=\Ebb_{Z}\nabla g(x+\sqrt{1-t}Z),~ d_m=\frac{\sum_{i=1}^m\nabla g(x+\sqrt{1-t}Z_i)}{m},\\
&e=\Ebb_{Z} g(x+\sqrt{1-t}Z),~ e_m=\frac{\sum_{i=1}^m g(x+\sqrt{1-t}Z_i)}{m},\\
&d_m^{'}=\frac{\sum_{i=1}^m\nabla g(x+\sqrt{1-t}Z_i^{'})}{m},~e_m^{'}=\frac{\sum_{i=1}^m g(x+\sqrt{1-t}Z_i^{'})}{m}.
\end{align*}
Due to $d-d_m=\Ebb\left[d_m^{'}-d_m|\bZ\right]$,
then $\|d-d_m\|_2^2\leq \Ebb\left[\|d_m^{'}-d_m\|_2^2|\bZ\right]$.
Since $g$ and $\nabla g$ are Lipschitz continuous, there exists a finite and positive constant $\gamma$
such that for all $x,y \in \mathbb{R}^p$,
\begin{align}\label{L1}
|g(x)-g(y)|\leq \gamma \|x-y\|_2,
\end{align}
\begin{align}\label{L2}
\|\nabla g(x)- \nabla g(y)\|_2\leq \gamma \|x-y\|_2.
\end{align}
Then,
\begin{align}\label{mm1}
\Ebb\|d-d_m\|_2^2
&\leq \Ebb\left[\Ebb[\|d_m^{'}-d_m\|_2^2|\bZ]\right]=\Ebb\|d_m^{'}-d_m\|_2^2\notag\\
&=\frac{\Ebb_{Z_1,Z_1^{'}}\left\|\nabla g(x+\sqrt{1-t}Z_1)-\nabla g(x+\sqrt{1-t}Z_1^{'})\right\|_2^2}{m}\notag\\
&\leq \frac{(1-t)\gamma^2}{m}\Ebb_{Z_1,Z_1^{'}}\left\|Z_1-Z_1^{'}\right\|_2^2\notag\\
&\leq \frac{2p\gamma^2}{m},
\end{align}
where the second inequality holds by \eqref{L2}.
Similarly, we also have
\begin{align}\label{mm2}
\Ebb|e-e_m|^2
&\leq \Ebb|e_m^{'}-e_m|^2\notag\\
&=\frac{\Ebb_{Z_1,Z_1^{'}}
\left|g(x+\sqrt{1-t}Z_1)-g(x+\sqrt{1-t}Z_1^{'})\right|^2}{m}\notag\\
&\leq \frac{(1-t)\gamma^2}{m}\Ebb_{Z_1,Z_1^{'}}\left\|Z_1-Z_1^{'}\right\|^2_2\notag\\
&\leq \frac{2p\gamma^2}{m},
\end{align}
where the second inequality holds due to \eqref{L1}.
By \eqref{mm1} and \eqref{mm2}, it follows that
\begin{align}\label{mm3}
\underset{x\in \mathbb{R}^p,t\in [0,1]}{\sup}\Ebb\left\|d-d_m\right\|_2^2
\leq \frac{2p\gamma^{2}}{m},
\end{align}
\begin{align}\label{mm4}
\underset{x\in \mathbb{R}^p,t\in [0,1]}{\sup}\Ebb|e-e_m|^2\leq \frac{2p\gamma^2}{m}.
\end{align}
Since  $g$ has a lower bound greater than 0, there exists a finite and positive constant $\xi$ such that
\begin{align}\label{boundf}
 g \geq \xi>0.
\end{align}
Then, by\eqref{L1}, \eqref{L2} and \eqref{boundf}, through some simple calculation, it yields that
\begin{align}\label{mm5}
\|b(x,t)-\tilde{b}_m(x,t)\|_2
&=\left\|\frac{d}{e}-\frac{d_m}{e_m}\right\|_2\notag\\
&\leq \frac{\|d\|_2|e_m-e|+\|d-d_m\|_2|e|}{|ee_m|}\notag\\
&\leq \frac{\gamma|e_m-e|+\|d-d_m\|_2|e|}{\xi^2}.
\end{align}
Let $R>0$, then
\begin{align}\label{mm6}
\underset{\|x\|_2\leq R}\sup g(x)\leq \mathcal{O}\left(\exp(R^2/2)\right).
\end{align}
Therefore, 
by \eqref{mm3}-\eqref{mm4} and \eqref{mm6},
it can be concluded  that
\begin{align*}
\underset{\|x\|_2 \leq R,t\in [0,1]}{\sup} \Ebb\left[\|b(x,t)-\tilde{b}_m(x,t)\|_2^2\right]
\leq 
\mathcal{O}\left(\frac{p\exp(R^2)}{m}\right).
\end{align*}

Moreover, if $g$ has a finite upper bound, that is, there exists a finite and positive constant $\zeta$ such that $g\leq\zeta$, then similar to \eqref{mm5}, it follows that for all $x\in \mathbb{R}^p$ and $t\in[0,1]$,
\begin{align}\label{mm55}
\|b(x,t)-\tilde{b}_m(x,t)\|_2
&\leq  \frac{\gamma|e_m-e|+\zeta\|d-d_m\|_2}{\xi^2}.
\end{align}
By \eqref{mm3}-\eqref{mm4} and \eqref{mm55}, we have
\begin{align*}
\underset{x\in\mathbb{R}^p,t\in [0,1]}{\sup} \Ebb\left[\|b(x,t)-\tilde{b}_m(t,x)\|_2^2\right]
\leq \mathcal{O}\left(\frac{p}{m}\right).
\end{align*}
\end{proof}

\begin{lemma}\label{lemma5}
Assume that $g$ is $\gamma$-Lipschitz continuous and has the lower bound greater than 0,
that is, $g\geq \xi >0$ for a positive and finite constant $\xi$, then for $k=0,1,\ldots,K$,
\begin{align*}
\Ebb\|\widetilde{Y}_{t_{k}}\|^2_2
\leq  \frac{6\gamma^2}{\xi^2}+3p.
\end{align*}
\end{lemma}
\begin{proof}
Define
$\Theta_{k,t}=\widetilde{Y}_{t_k}+(t-t_k)\tilde{b}_m(\widetilde{Y}_{t_k},t_k)$
and $\widetilde{Y}_{t}=\Theta_{k,t}+B_t-B_{t_k}$,
where $t_k \leq t \leq t_{k+1}$ with $k=0,1,\ldots,K-1$.
Since $g$ is $\gamma$-Lipschitz continuous and $g\geq \xi >0$, for all $x \in \mathbb{R}^p$ and $t\in[0,1]$,
\begin{align}\label{eq2}
\|b(x,t)\|_2^2\leq \frac{\gamma^2}{\xi^2},~~ \|\tilde{b}_m(x,t)\|_2^2 \leq \frac{\gamma^2}{\xi^2}.
\end{align}
By \eqref{eq2},  we have
\begin{align*}
\|\Theta_{k,t}\|^2_2&=\|\widetilde{Y}_{t_k}\|^2_2+(t-t_k)^2\|\tilde{b}_m(\widetilde{Y}_{t_k},t_k)\|^2_2
+2(t-t_k)\widetilde{Y}_{t_k}^{\top}\tilde{b}_m(\widetilde{Y}_{t_k},t_k)\\
&\leq (1+s)\|\widetilde{Y}_{t_k}\|^2_2+\frac{(s+s^2)\gamma^2}{\xi^2}.
\end{align*}
Furthermore, we have
\begin{align*}
\Ebb\left[\|\widetilde{Y}_{t}\|^2_2|\widetilde{Y}_{t_k}\right]&=\Ebb\left[\|\Theta_{k,t}\|^2_2|\widetilde{Y}_{t_k}\right]+(t-t_k)p\\
&\leq(1+s)\Ebb\|\widetilde{Y}_{t_k}\|^2_2+\frac{(s+s^2)\gamma^2}{\xi^2}+sp.
\end{align*}
Therefore,
\begin{align*}
\Ebb\|\widetilde{Y}_{t_{k+1}}\|^2_2
\leq(1+s)\Ebb\|\widetilde{Y}_{t_k}\|^2_2+\frac{(s+s^2)\gamma^2}{\xi^2}+sp.
\end{align*}
Since $\widetilde{Y}_{t_0}=0$, by induction, we have
\begin{align*}
\Ebb\|\widetilde{Y}_{t_{k+1}}\|^2_2
\leq \frac{6\gamma^2}{\xi^2}+3p.
\end{align*}
\end{proof}
\begin{lemma}\label{lemma6}
If $g$  and $\nabla g$ are Lipschitz continuous, and $g$ has the lower bound greater than 0, then for $k=0,1,\ldots,K$ and $t\in [0,1]$,
\begin{align*}
\Ebb\left\|b(\widetilde{Y}_{t_k},t_k)-\tilde{b}_m(\widetilde{Y}_{t_k},t_k)\right\|_2^2\leq
\mathcal{O}\left(\frac{p}{\log(m)}\right),
\end{align*}
\begin{align*}
\Ebb\left\|b(X_{t},t)-\tilde{b}_m(X_{t},t)\right\|_2^2\leq
\mathcal{O}\left(\frac{p}{\log(m)}\right).
\end{align*}
Moreover, if $g$  has a finite upper bound, then
\begin{align*}
\Ebb\left\|b(\widetilde{Y}_{t_k},t_k)-\tilde{b}_m(\widetilde{Y}_{t_k},t_k)\right\|_2^2\leq
\mathcal{O}\left(\frac{p}{m}\right),
\end{align*}
\begin{align*}
\Ebb\left\|b(X_{t},t)-\tilde{b}_m(X_{t},t)\right\|_2^2\leq
\mathcal{O}\left(\frac{p}{m}\right).
\end{align*}
\end{lemma}
\begin{proof}
Let $R>0$, then
\begin{equation}\label{erro1}
\begin{split}
\Ebb\left\|b(\widetilde{Y}_{t_k},t_k)-\tilde{b}_m(\widetilde{Y}_{t_k},t_k)\right\|_2^2
&=\Ebb_{\widetilde{Y}_{t_k}}\Ebb_Z\left[\left\|b(\widetilde{Y}_{t_k},t_k)-\tilde{b}_m(\widetilde{Y}_{t_k},t_k)\right\|_2^21(\|\widetilde{Y}_{t_k}\|_2\leq R)\right]\\
&~~~+\Ebb_{\widetilde{Y}_{t_k}}\Ebb_Z\left[\left\|b(\widetilde{Y}_{t_k},t_k)-\tilde{b}_m(\widetilde{Y}_{t_k},t_k)\right\|_2^21(\|\widetilde{Y}_{t_k}\|_2> R)\right].
\end{split}
\end{equation}
Next, we need to bound the two terms of \eqref{erro1}.
First, by Lemma \ref{lemma4}, we have
$$
\Ebb_{\widetilde{Y}_{t_k}}\Ebb_Z\left[\left\|b(\widetilde{Y}_{t_k},t_k)-\tilde{b}_m(\widetilde{Y}_{t_k},t_k)\right\|_2^21(\|\widetilde{Y}_{t_k}\|_2\leq R)\right]\leq
\mathcal{O}\left(\frac{p\exp(R^2)}{m}\right).
$$
Secondly, combining \eqref{eq2} and Lemma \eqref{lemma5} with Markov inequality,
$$
\Ebb_{\widetilde{Y}_{t_k}}\Ebb_Z\left[\left\|b(\widetilde{Y}_{t_k},t_k)-\tilde{b}_m(\widetilde{Y}_{t_k},t_k)\right\|_2^21(\|\widetilde{Y}_{t_k}\|_2> R)\right]
\leq \mathcal{O}\left(p/R^2\right).
$$
Hence
\begin{align}\label{err1}
\Ebb\left\|b(\widetilde{Y}_{t_k},t_k)-\tilde{b}_m(\widetilde{Y}_{t_k},t_k)\right\|_2^2\leq
\mathcal{O}\left(\frac{p\exp(R^2)}{m}\right)
+\mathcal{O}\left(p/R^2\right).
\end{align}
Similar to \eqref{err1}, by Lemma \ref{lemma1o} and Lemma \ref{lemma4}, we have
\begin{align}\label{berr1}
\Ebb\left\|b(X_{t},t)-\tilde{b}_m(X_{t},t)\right\|_2^2\leq
\mathcal{O}\left(\frac{p\exp(R^2)}{m}\right)
+\mathcal{O}\left(p/R^2\right).
\end{align}
Set $R=\left(\frac{\log(m)}{2}\right)^{1/2}$ in \eqref{err1} and \eqref{berr1}, then it yields that
\begin{align*}
\Ebb\left\|b(\widetilde{Y}_{t_k},t_k)-\tilde{b}_m(\widetilde{Y}_{t_k},t_k)\right\|_2^2\leq
\mathcal{O}\left(p/\log(m)\right),
\end{align*}
\begin{align*}
\Ebb\left\|b(X_{t},t)-\tilde{b}_m(X_{t},t)\right\|_2^2\leq
\mathcal{O}\left(p/\log(m)\right).
\end{align*}

Moreover, if $g$ has a finite upper bound, by Lemma \ref{lemma4}, we can similarly get
$$
\Ebb\left\|b(\widetilde{Y}_{t_k},t_k)-\tilde{b}_m(\widetilde{Y}_{t_k},t_k)\right\|_2^2=
\Ebb_{\widetilde{Y}_{t_k}}\Ebb_{Z}\left[\left\|b(\widetilde{Y}_{t_k},t_k)-\tilde{b}_m(\widetilde{Y}_{t_k},t_k)\right\|_2^2\right]
\leq \mathcal{O}\left(p/m\right),
$$
$$
\Ebb\left\|b(X_{t},t)-\tilde{b}_m(X_{t},t)\right\|_2^2\leq \mathcal{O}\left(p/m\right).
$$
This completes the proof.
\end{proof}
\subsection{Proof of Theorem \ref{th2}}
\begin{proof}
Let $\Delta_k=X_{t_k}-\widetilde{Y}_{t_k}$.
Then,
\begin{align*}
\Delta_{k+1}&=\Delta_{k}+(X_{t_{k+1}}-X_{t_k})-(\widetilde{Y}_{t_{k+1}}-\widetilde{Y}_{t_k})\\
&=\Delta_{k}-s\left[\tilde{b}_m(\widetilde{Y}_{t_k},t_k)-\tilde{b}_m(\widetilde{Y}_{t_k}+
\Delta_{k},t_k)\right] +\int_{t_k}^{t_{k+1}}\left[b(X_t,t)-\tilde{b}_m(X_{t_k},t_k)\right]\mathrm{d}t.
\end{align*}
Combining Lemma \ref{lemma3} with the triangle inequality,
we have
\begin{align*}
&\left\|\Delta_{k}-s\left[\tilde{b}_m(\widetilde{Y}_{t_k},t_k)-\tilde{b}_m(\widetilde{Y}_{t_k}
+\Delta_{k},t_k)\right]\right\|_{L_2}\\
&\leq
\left\|\Delta_{k}-s\left[b(\widetilde{Y}_{t_k},t_k)-b(\widetilde{Y}_{t_k}+
\Delta_{k},t_k)\right]\right\|_{L_2}
+s\left\|b(\widetilde{Y}_{t_k},t_k)-\tilde{b}_m(\widetilde{Y}_{t_k},t_k)\right\|_{L_2}\\
&~~~+s\left\|\tilde{b}_m(\widetilde{Y}_{t_k}+\Delta_{k},t_k)-b(\widetilde{Y}_{t_k}+\Delta_{k},t_k)\right\|_{L_2}\\
&\leq \rho \|\Delta_{k}\|_{L_2}+s\left\|b(\widetilde{Y}_{t_k},t_k)-\tilde{b}_m(\widetilde{Y}_{t_k},t_k)\right\|_{L_2}
+s\left\|b(\widetilde{Y}_{t_k}+\Delta_{k},t_k)-\tilde{b}_m(\widetilde{Y}_{t_k}+\Delta_{k},t_k)\right\|_{L_2}.
\end{align*}
Therefore, by Lemma \ref{lemma6},
\begin{align}\label{eqq1}
\left\|\Delta_{k}-s\left[\tilde{b}_m(\widetilde{Y}_{t_k},t_k)-\tilde{b}_m(\widetilde{Y}_{t_k}+\Delta_{k},t_k)\right]\right\|_{L_2}
\leq \rho \|\Delta_{k}\|_{L_2}+s\cdot
\mathcal{O}\left(\sqrt{\frac{p}{\log(m)}}\right).
\end{align}
Moreover,
\begin{align*}
&\left\|\int_{t_k}^{t_{k+1}}\left[b(X_t,t)-\tilde{b}_m(X_{t_k},t_k)\right]\mathrm{d}t\right\|_2^2\\
&\leq s\int_{t_k}^{t_{k+1}}\left\|b(X_t,t)-\tilde{b}_m(X_{t_k},t_k)\right\|_2^2\mathrm{d}t\\
&\leq 2s\int_{t_k}^{t_{k+1}}\left\|b(X_t,t)-b(X_{t_k},t_k)\right\|_2^2\mathrm{d}t
+2s^2\left\|b(X_{t_k},t_k)-\tilde{b}_m(X_{t_k},t_k)\right\|_2^2\\
&\leq 4sC_1^2\int_{t_k}^{t_{k+1}}\left[\|X_t-X_{t_k}\|^2+|t-t_k|\right]\mathrm{d}t
+2s^2\left\|b(X_{t_k},t_k)-\tilde{b}_m(X_{t_k},t_k)\right\|_2^2\\
&\leq 4sC_1^2\int_{t_k}^{t_{k+1}}\|X_t-X_{t_k}\|^2\mathrm{d}t+4C_1^2 s^3
+2s^2\left\|b(X_{t_k},t_k)-\tilde{b}_m(X_{t_k},t_k)\right\|_2^2,
\end{align*}
where the third inequality holds by \eqref{cond3}.
Furthermore, by Lemmas \ref{lemma2o} and \ref{lemma6}, we have
\begin{align*}
&\Ebb\left[\left\|\int_{t_k}^{t_{k+1}}(b(X_t,t)-\tilde{b}_m(X_{t_k},t_k))\mathrm{d}t\right\|_2^2\right]\\
&\leq
4sC_1^2\int_{t_k}^{t_{k+1}}\Ebb\|X_t-X_{t_k}\|^2\mathrm{d}t+4C_1^2 s^3
+2s^2\Ebb\left[\|b(X_{t_k},t_k)-\tilde{b}_m(X_{t_k},t_k)\|_2^2\right]\\
&\leq
4s^2C_1^2[4C_0\exp(2C_0)(C_0+p)s^2+2C_0s^2+2ps]
+4C_1^2s^3+2s^2\Ebb\left[\|b(X_{t_k},t_k)-\tilde{b}_m(X_{t_k},t_k)\|_2^2\right]\\
&\leq
4s^2C_1^2[4C_0\exp(2C_0)(C_0+p)s^2+2C_0s^2+2ps]
+4C_1^2s^3+s^2\mathcal{O}\left(\frac{p}{\log(m)}\right).
\end{align*}
Therefore,
\begin{align}\label{eqq2}
\left\|\int_{t_k}^{t_{k+1}}(b(X_t,t)-\tilde{b}_m(X_{t_k},t_k))\mathrm{d}t\right\|_{L_2}
\leq
\mathcal{O}(\sqrt{p}s^{3/2})+\mathcal{O}\left(s\sqrt{\frac{p}{\log(m)}}\right).
\end{align}
By \eqref{eqq1} and \eqref{eqq2}, we have
\begin{align*}
\|\Delta_{k+1}\|_{L_2}\leq \rho \|\Delta_{k}\|_{L_2}+
\mathcal{O}(\sqrt{p}s^{3/2})+\mathcal{O}\left(s\sqrt{\frac{p}{\log(m)}}\right).
\end{align*}
From the definition of $\rho$ in Lemma \ref{lemma3}, we can get
\begin{align*}
\|\Delta_{k+1}\|_{L_2}\leq \rho^{k+1}\|\Delta_0\|_{L_2}+\mathcal{O}(\sqrt{sp})+\mathcal{O}\left(\sqrt{\frac{p}{\log(m)}}\right).
\end{align*}
Furthermore, since we set $X_{t_0}=\widetilde{Y}_{t_0}=0$, we have
\begin{align*}
W_2(\mbox{Law}(\widetilde{Y}_{t_K}),\mu)&\leq \rho^{K}\|\Delta_0\|_{L_2}+\mathcal{O}(\sqrt{sp})+\mathcal{O}\left(\sqrt{\frac{p}{\log(m)}}\right)\\
&\leq \mathcal{O}(\sqrt{sp})+\mathcal{O}\left(\sqrt{\frac{p}{\log(m)}}\right).
\end{align*}
\end{proof}
\subsection{Proof of Theorem \ref{th3}}
\begin{proof}
This proof is same as that of Theorem \ref{th2}.
Similar to \eqref{eqq1}, by Lemma \ref{lemma6}, we have
\begin{align}\label{th2eq1}
&\left\|\Delta_{k}-s\left[\tilde{b}_m(\widetilde{Y}_{t_k},t_k)-\tilde{b}_m(\widetilde{Y}_{t_k}
+\Delta_{k},t_k)\right]\right\|_{L_2} \notag \\
&\leq \rho \|\Delta_{k}\|_{L_2}
+s\left\|b(\widetilde{Y}_{t_k},t_k)-\tilde{b}_m(\widetilde{Y}_{t_k},t_k)\right\|_{L_2}
+s\left\|b(\widetilde{Y}_{t_k}+\Delta_{k},t_k)-\tilde{b}_m(\widetilde{Y}_{t_k}+\Delta_{k},t_k)\right\|_{L_2} \notag\\
&\leq \rho \|\Delta_{k}\|_{L_2}+s\cdot
\mathcal{O}\left(\sqrt{\frac{p}{m}}\right).
\end{align}
Then we also have
\begin{align}\label{th2eq2}
\left\|\int_{t_k}^{t_{k+1}}(b(X_t,t)-\tilde{b}_m(X_{t_k},t_k))\mathrm{d}t\right\|_{L_2}
\leq
\mathcal{O}(\sqrt{p}s^{3/2})+\mathcal{O}\left(s\sqrt{\frac{p}{m}}\right).
\end{align}
By \eqref{th2eq1}  and \eqref{th2eq2},
\begin{align*}
\|\Delta_{k+1}\|_{L_2}\leq \rho \|\Delta_{k}\|_{L_2}+
\mathcal{O}(\sqrt{p}s^{3/2})+\mathcal{O}\left(s\sqrt{\frac{p}{m}}\right).
\end{align*}
Therefore, by induction and $X_{t_0}=\widetilde{Y}_{t_0}=0$, it follows that
\begin{align*}
W_2(\mbox{Law}(\widetilde{Y}_{t_K}),\mu)\leq \mathcal{O}(\sqrt{ps})+\mathcal{O}\left(\sqrt{\frac{p}{m}}\right).
\end{align*}
\end{proof}
\subsection{Proof of Theorem \ref{th4}}
\begin{proof}
By triangle inequality, we have  $$W_2(Law(Y_{t_K}(\varepsilon)),\mu)\leq
W_2(\mu,\mu_{\varepsilon})+W_2(Law(Y_{t_K}(\varepsilon)),\mu_{\varepsilon})
$$
and
$$W_2(Law(\widetilde{Y}_{t_K}(\varepsilon)),\mu)\leq
W_2(\mu,\mu_{\varepsilon})+W_2(Law(\widetilde{Y}_{t_K}(\varepsilon)),\mu_{\varepsilon}).
$$
First, we show that $W_2(\mu,\mu_{\varepsilon})$ will converge to zero as $\varepsilon$ goes to zero.
Let $Y\sim \mu$ and $Z\sim N(0,\bI_p)$,  and let $\tau$ be a  Bernoulli random variable with  $P(\tau=1)=1-\varepsilon$ and
$P(\tau=0)=\varepsilon$. Assume $Y$, $Z$ and $\tau$ are mutually independent.
Then $(Y,(1-\tau)Z+\tau Y)$ is a coupling of $(\mu,\mu_{\varepsilon})$. Denote the joint distribution
of $(Y,(1-\tau)Z+\tau Y)$ by $\pi$. Then, we have
\begin{align*}
\int_{\mathbb{R}^p\times\mathbb{R}^p}\|x-y\|_2^2d \pi&=\Ebb\left\|Y-((1-\tau)Z+\tau Y)\right\|_2^2\\
&=\Ebb\left[\Ebb\left[\|Y-((1-\tau)Z+\tau Y)\|_2^2|\tau\right]\right]\\
&=\Ebb\left[\Ebb\left[\|Y-((1-\tau)Z+\tau Y)\|_2^2|\tau=1\right]\right]P(\tau=1)\\
&~~~+\Ebb\left[\Ebb\left[\|Y-((1-\tau)Z+\tau Y)\|_2^2|\tau=0\right]\right]P(\tau=0)\\
&=\Ebb[\|Y-Z\|_2^2|\tau=0]P(\tau=0)\\
&=\varepsilon \Ebb\|Y-Z\|_2^2\\
&=\mathcal{O}(p\varepsilon).
\end{align*}
Therefore, it follows that
\begin{align}\label{W2muo}
\underset{\varepsilon\rightarrow 0}{\lim}~W_2(\mu,\mu_{\varepsilon})=0.
\end{align}
Similar to the proof of Theorems \ref{th1}-\ref{th2},  we have
\begin{align}\label{ers1}
\underset{K \rightarrow \infty}{\lim}~ W_2(Law(Y_{t_K}(\varepsilon)),\mu_{\varepsilon})=0,
\end{align}
\begin{align}\label{ers2}
\underset{m,K \rightarrow \infty}{\lim}~ W_2(Law(\widetilde{Y}_{t_K}(\varepsilon)),\mu_{\varepsilon})=0.
\end{align}
Combining \eqref{W2muo} with \eqref{ers1} and \eqref{ers2}, we have
\begin{align*}
&\underset{K \rightarrow \infty,\varepsilon \rightarrow 0}{\lim} W_2(\mbox{Law}(Y_{t_K}(\varepsilon)),\mu)=0,\\
&\underset{m,K \rightarrow \infty, \varepsilon \rightarrow 0}{\lim} W_2(\mbox{Law}(\widetilde{Y}_{t_K}(\varepsilon)),\mu)=0.
\end{align*}
This completes the proof.
\end{proof}

\end{appendix}
{\singlespace
\bibliography{schrodinger_bib}   
}

\end{document}